\documentclass[12pt]{article}

\usepackage{header}

\title{\Large{ {\bf Local Strategy-proofness and Dictatorship }}\thanks{We are grateful to Anirban Kar, Abhimanyu Khan, Debasis Mishra, Nozomu Muto, Evan Piermont, Soumendu Sarkar, Arunava Sen and Huaxia Zeng for their insightful discussions and comments at various stages of this work. We also thank seminar and conference participants at IIT Jodhpur, IISER Bhopal, ISI Delhi, ISI Tezpur, Hitotsubashi University, Delhi School of Economics, ISI Kolkata, and the University of Essex for their valuable comments and suggestions. We also thank Claude for assistance with editing. All remaining errors are our own.}}

\author{Abinash Panda\thanks{Department of Economics, Shiv Nadar Institution of Eminence, NH - 91, Gautam Buddha Nagar, Uttar Pradesh, India - 201314. Email: \texttt{ap280@snu.edu.in}}, Anup Pramanik\thanks{Department of Economics, Shiv Nadar Institution of Eminence, NH - 91, Gautam Buddha Nagar, Uttar Pradesh, India - 201314. Email: \texttt{anup.pramanik@snu.edu.in}.} and Ragini Saxena\thanks{Department of Economics, University of Rochester, Rochester, NY 14627, United States. Email: \texttt{rsaxena6@ur.rochester.edu}.}} 

\begin{document}

\maketitle

\begin{abstract} We investigate preference domains on which every unanimous and locally strategy-proof social choice function (scf) is dictatorial. We identify a condition on domains, \textit{connected with distinct neighbours}, that is necessary for this property. We then show that the condition is also sufficient on \textit{L-tops-only} domains, i.e., domains on which every unanimous and locally strategy-proof scf is tops-only. Whether every connected with distinct neighbours domain is L-tops-only remains open; if some are not, our sufficiency result is silent on them. We show, however, that every connected with distinct neighbours domain, L-tops-only or not, is a \textit{dictatorial domain}, i.e., unanimity and strategy-proofness (a stronger requirement than local strategy-proofness) force dictatorship on it.
\end{abstract}



\newpage

\section{Introduction}

The theory of mechanism design is fundamental to understanding how collective decisions are made based on private information held by individual agents. A central focus in the literature has been the study of mechanisms that are strategy-proof (or dominant strategy incentive-compatible), meaning that truthfully revealing one’s preferences is a weakly dominant strategy for every agent. Strategy-proofness is a stringent requirement with strong implications. In particular, in the context of voting, the Gibbard-Satterthwaite Theorem (\cite{gibbard1973manipulation} and \cite{satterthwaite1975strategy}) establishes that, under unrestricted preferences, the only social choice functions (scfs) that are both strategy-proof and unanimous (a weak efficiency condition) are dictatorial.

A refinement of this framework is local strategy-proofness, which relaxes the standard requirement that agents must have no incentive to misreport their preferences across the entire preference domain. Instead, it requires incentive compatibility only with respect to a restricted set of nearby, or ``local" preferences. This restriction is motivated by practical considerations: in many real-world settings, agents may be unwilling or unable to report preferences that differ substantially from their true ones, due to cognitive limitations, informational constraints, or institutional barriers.\footnote{See, respectively, \cite{conlisk1996bounded} on bounded
rationality, \cite{mackowiak2023rational} on limited attention and
information processing, and \cite{haeringer2009constrained} on
institutional limits on preference reporting in school choice.}  The class of locally strategy-proof scfs is often thought to be broader than the class of  strategy-proof scfs. However, earlier work by \cite{Carroll12} and \cite{sato2013sufficient} demonstrates that for many important preference domains, the two notions coincide when localness is defined naturally, such as through adjacency (a swap of two consecutively ranked alternatives). This phenomenon, known as local-global equivalence, has significant theoretical and practical implications for mechanism design.

This paper examines the standard voting framework, in which a finite set of agents (or voters) holds strict preferences over a finite set of alternatives. These preferences are private information, and a scf selects a single alternative for each profile of reported preferences. Following \cite{sato2013sufficient}, we adopt a natural notion of localness: one preference ordering is said to be adjacent (or local) to another if it differs by a single swap of two consecutively ranked alternatives.

Our study is motivated by a fundamental asymmetry between these two notions of strategy-proofness. If local strategy-proofness and unanimity imply dictatorship on a domain, then strategy-proofness and unanimity, being stronger requirements, must also imply dictatorship on that domain. However, the converse does not hold: there exist many important domains where strategy-proofness and unanimity imply dictatorship, but local strategy-proofness and unanimity do not. For instance, such behaviour arises in \textit{circular domains} discussed in \cite{sato2010circular} and in several important domains analysed in \cite{Aswal03}. This
opens the door to designing desirable, non-dictatorial scfs that satisfy local strategy-proofness.

One might view local strategy-proofness merely as a technical relaxation of strategy-proofness, of interest mainly because it can be easier to check and, on many domains, implies the stronger property. Our motivation is different. When agents face cognitive, informational, or institutional limits on how far their reports can depart from their true preferences, an agent may simply be unable to report a preference that is far from her true one ---
so the only misreports she can actually make are to nearby preferences. In that case local strategy-proofness is not a weaker stand-in for strategy-proofness; it directly captures the incentive constraint the agent actually faces. On such domains, requiring strategy-proofness imposes a stronger
incentive standard than agents' actual reporting constraints call for. This makes the following question worth asking on its own terms, separately from the local-global equivalence question: on domains where local strategy-proofness is the relevant incentive
requirement, does unanimity force dictatorship, or does this weaker requirement leave room for non-dictatorial rules? This is the question we study, and answering it clarifies how far incentive constraints can be weakened before the space of implementable social choice functions genuinely expands.

Throughout, we restrict attention to minimally rich domains, i.e., domains in which each alternative is top-ranked in at least one preference ordering in the domain. Our main results are as follows:

\begin{enumerate}
\item Necessary Condition: We introduce the notion of the \textit{connected with distinct neighbours} property and establish that any domain on which every unanimous and locally strategy-proof scf is dictatorial must satisfy this property (Theorem \ref{thm:1}).

\item Sufficient Condition: We show that this necessary condition is also sufficient for a broad class of domains satisfying an additional restriction, which we call \textit{L-tops-only} domains --- those in which every unanimous and locally strategy-proof scf satisfies tops-onlyness (Theorem \ref{thm:2}).

\item Strengthening the Sufficiency Result Beyond L-tops-only Domains: It remains open whether every connected with distinct neighbours domain is L-tops-only. If some are not, Theorem \ref{thm:2} is silent on them. A domain is a \textit{dictatorial domain} if every unanimous and strategy-proof scf defined on it is dictatorial (\cite{Aswal03}). Theorem \ref{thm:3} shows, however, that every connected with distinct neighbours domain, L-tops-only or not, lies in this class. In exchange for dropping L-tops-only, Theorem \ref{thm:3} requires the stronger incentive property of strategy-proofness in place of local strategy-proofness.
\end{enumerate}

Informally, a domain is said to be connected with distinct neighbours if it satisfies two key structural properties. First, the domain must be connected, meaning that any two preferences in the domain can be linked through a sequence of adjacent preferences. Second, and more critically, the domain must exhibit a form of local diversity around each preference. For any given preference, consider the set of preferences that share the same top-ranked alternative and are connected through paths that preserve this top alternative throughout --- this set is referred to as the top-connected closure of the preference. The requirement is that this top-connected closure must have at least two neighbouring preferences (i.e., preferences that are adjacent to this set) whose top-ranked alternatives are distinct. It is worth noting that well-known domains such as the single-peaked and single-crossing domains do not satisfy this property. However, the union of single-peaked and single-dipped domains (with respect to the same linear order over alternatives) does satisfy it, as does the unrestricted domain.

The work closest to ours is \cite{hong2023unanimity}, who study the problem of identifying domains on which unanimity and local strategy-proofness imply dictatorship.\footnote{\cite{hong2023unanimity}'s result also applies to the local version of ordinal Bayesian incentive compatibility.} They focus on the class of \textit{sparsely connected domains without restoration} (SCD) and characterize domains within this class using the \textit{disagreement property}. Specifically, they show that the disagreement property is both necessary and sufficient for unanimity and local strategy-proofness to imply dictatorship on SCD domains. Additionally, they establish that connectedness combined with the disagreement property forms a necessary condition for dictatorship under unanimity and local strategy-proofness. In contrast, our work identifies a stronger necessary condition, the connected with distinct neighbours property, which implies both connectedness and the disagreement property (see Remark \ref{R7}). Moreover, we provide a counterexample (Example \ref{ex7}) showing that the implication is strict: there exist domains that satisfy connectedness and the disagreement property but not the connected with distinct neighbours property, thereby demonstrating that our condition is indeed strictly stronger. Since SCD domains are L-tops-only, our results also apply within the SCD class, where the connected with distinct neighbours property and the disagreement property turn out to be equivalent. Our sufficient condition, by contrast, applies to a broader class of domains and does not require the SCD structure. More broadly, our approach focuses on identifying structural conditions on domains rather than working within a fixed class such as SCD domains. We elaborate on these relationships in Section \ref{sec5}.

Our work offers new insight into the literature on local-global equivalence. This line of research in the voting framework begins with \cite{sato2013sufficient}, who provide a sufficient condition on domains under which local strategy-proofness implies strategy-proofness. Subsequently, \cite{kumar2021local} formulate the local-global equivalence problem more generally within the context of an environment, represented as a graph where nodes correspond to admissible preferences and edges encode the notion of localness.  They characterize environments in which local strategy-proofness implies strategy-proofness.\footnote{See \cite{cho2025local} and \cite{kumar2025equivalence} for recent generalizations to directed networks and broader domain structures.} \footnote{Local-global equivalence has also been studied in models with monetary transfers and quasilinear preferences (see, e.g., \cite{archer2008truthful}, \cite{Carroll12}, \cite{archer2014truthful}, \cite{mishra2016local} and \cite{kumar2024local}).} However, when restricting attention to unanimous scfs, such a characterization does not exist. \cite{kumar2021localb} and \cite{hong2023unanimity} provide sufficient conditions under which local strategy-proofness implies strategy-proofness for unanimous scfs, but these conditions are not necessary. In particular, \cite{kumar2021localb} shows that domains satisfying the \textit{pairwise no-restoration property} (Property P) ensure local-global equivalence for unanimous scfs, while \cite{hong2023unanimity} establish a similar result for SCD domains.  The SCD condition is slightly weaker than Property P, and on minimally rich domains, the two coincide. Our Theorem \ref{thm:2} provides a sufficiency result that is independent of both Property P and SCD, which we discuss in detail in Section \ref{sec5}.

Our work also contributes to the broader literature on \textit{dictatorial domains}, i.e., preference domains where strategy-proofness and unanimity imply dictatorship (see, for example, \cite{Aswal03}, \cite{sato2010circular}, \cite{pramanik2015further}, \cite{chatterji2023taxonomy}, among others). Since our Theorem \ref{thm:3} establishes that domains satisfying the connected with distinct neighbors property are dictatorial, our findings contribute to this literature. Moreover, dictatorial scfs satisfy  tops-only property. Theorem \ref{thm:3} demonstrate that connected with distinct neighbours domains are tops-only, providing new insights into the literature on tops-only domains (see, for example, \cite{Weymark08}, \cite{Chatt11}, and \cite{chatterji2018random}).

The remainder of the paper is organized as follows. Section \ref{sec2} introduces the model and key definitions. Section \ref{sec3} presents our main results, including a necessary condition and a sufficient condition under which unanimity and local strategy-proofness imply dictatorship. Section \ref{sec4} shows that on connected with distinct neighbours domains, unanimity and strategy-proofness (a stronger requirement than local strategy-proofness) force dictatorship. Section \ref{sec5} discusses how our results contribute to the existing literature and outlines directions for future research. Finally, Section \ref{sec6} concludes the paper.

\section{The Framework} \label{sec2}
Let $A$ be a finite set of alternatives with $|A|=m\geq 3$ and $N=\{1,2,\ldots,n\}$ be a finite set of voters with $n\geq 2$. Each voter $i\in N$ has a preference $P_i$ over $A$, which is assumed to be a \emph{linear order} (i.e., an antisymmetric, complete, and transitive binary relation, where $a\;P_i\;b$ is interpreted as $a$ being strictly preferred to $b$ according to $P_i$).  Let $\mathbb{P}$ denote the set of all such linear orders over $A$, which we refer to as the \emph{unrestricted domain}. A \emph{domain} is any subset $\mathbb{D} \subseteq \mathbb{P}$; we assume all voters share the same domain $\mathbb{D}$. A \emph{preference profile} $P = (P_1, \dots, P_n) \in \mathbb{D}^n$ specifies the preferences of all voters. 

For any $P_i\in \mathbb{P}$ and $k \in\{1,\ldots,m\}$, let $r_k(P_i)$ denote the $k^{th}$ ranked alternative in $P_i$, i.e., $|\{a\in A: a\;P_i\;r_k(P_i)\}|=k-1$. For any $P_i\in \mathbb{P}$ and $a\in A$, let $r(P_i,a)$ denote its rank in $P_i$. Note that for any $P_i\in \mathbb{P}$, $k\in \{1,2,\ldots,m\}$ and $a\in A$, $r_k(P_i)=a$ if and only if $r(P_i,a)=k$. Two linear orders $P_i$ and $P'_i$ are called \emph{adjacent} (denoted $P_i \sim P'_i$) if there exist $a, b \in A$ such that $r(P_i, a) =r(P_i, b) +1$, $r(P'_i, a) =r(P_i, b)$, $r(P'_i, b) =r(P_i, a)$ and for all $c\in A \setminus \{a, b\}$, $r(P_i, c) =r(P'_i, c)$.

\begin{defn}\rm A \emph{social choice function} (scf) $f$ is a mapping from $\mathbb{D}^n$ to $A$, i.e., $f: \mathbb{D}^n\rightarrow A$.
\end{defn}

We now describe several key properties of scfs that are central to this study.

The first is \emph{unanimity}, which requires that if all voters rank a particular alternative $a$ as their top choice, then the function must select $a$. 

\begin{defn}\rm A scf $f: \mathbb{D}^n\rightarrow A$ satisfies \emph{unanimity} if for any $a\in A$ and $P\in \mathbb{D}^n$ with $r(P_i, a)=1$ for all $i\in N$, we have $f(P)=a$.
\end{defn}

Next, we consider two notions of strategy-proofness. A scf is \emph{locally strategy-proof} if a voter cannot benefit by misreporting to a preference adjacent to their true preference. In contrast, a scf is \emph{strategy-proof} if a voter cannot benefit from any form of misrepresentation. Note that local strategy-proofness is weaker than strategy-proofness.

\begin{defn}\rm A scf $f: \mathbb{D}^n\rightarrow A$ is \emph{locally manipulable} by a voter $i\in N$ at profile $P = (P_i,P_{-i})$ if there exists $P'_i\in \mathbb{D}$ with $P_i\sim P'_i$ such that $f(P'_i,P_{-i})\;P_i\;f(P_i,P_{-i})$. The scf $f$ is \emph{locally strategy-proof} if it is not locally manipulable by any voter at any profile.
\end{defn}

\begin{defn}\rm A scf $f: \mathbb{D}^n\rightarrow A$ is \emph{manipulable} by a voter $i\in N$ at profile $P = (P_i,P_{-i})$ if there exists $P'_i\in \mathbb{D}$ such that $f(P'_i,P_{-i})\;P_i\;f(P_i,P_{-i})$. The scf $f$ is \emph{strategy-proof} if it is not manipulable by any voter at any profile.
\end{defn}

The final property we consider is dictatorship, which plays a central role in this paper. Dictatorship requires that the outcome of the scf always coincides with the top-ranked alternative of some fixed voter, regardless of the preferences of others.

\begin{defn}\rm  A scf $f: \mathbb{D}^n\rightarrow A$ satisfies \emph{dictatorship} if there exists a voter $i\in N$ such that for any $P\in \mathbb{D}^n$, $f(P)= r_1(P_i)$.
\end{defn}

In this paper, we focus on domains that satisfy a basic richness condition: for every alternative, there exists at least one preference ordering in the domain where that alternative is ranked first.

\begin{defn}\rm A domain $\mathbb{D}\subseteq \mathbb{P}$ is \emph{minimally rich} if for any $a\in A$ there exists a linear order $P_i\in \mathbb{D}$ such that $r_1(P_i)=a$.
\end{defn}

Let $\mathbb{D}^{DICT}$ denote the set of minimally rich domains on which every unanimous and strategy-proof scf satisfies dictatorship; we refer to these as dictatorial domains. Similarly, let $\mathbb{D}^{LDICT}$ be the set of minimally rich domains where every unanimous and locally strategy-proof scf satisfies dictatorship. Since local strategy-proofness is a weaker requirement than strategy-proofness, it follows that any domain in $\mathbb{D}^{LDICT}$ must also belong to $\mathbb{D}^{DICT}$. The following example demonstrates that this inclusion is strict --- that is, $\mathbb{D}^{LDICT} \subset \mathbb{D}^{DICT}$.

\begin{example}\label{ex1}\rm Let $A=\{a_1,a_2,a_3,a_4\}$. The domain $\mathbb{D}=\{P_i^1,P_i^2,P_i^3,P_i^4,P_i^5,P_i^6,P_i^7,P_i^8\}$ is specified in Table \ref{T1}. $\mathbb{D}$ is called circular domain and it is a dictatorial domain (see, \cite{sato2010circular} and \cite{Chatt11} ). An important property of $\mathbb{D}$ is that no two preferences in $\mathbb{D}$ are adjacent. Therefore, any scf defined on $\mathbb{D}$ trivially satisfies local strategy-proofness, including those satisfying desirable properties such as unanimity, anonymity and non-dictatorship. \hfill $\blacksquare$

\begin{table}[h!]
\centering

\begin{tabular}{ c c c c c c c c }
				$P_i^1$ & $P_i^2$ & $P_i^3$ & $P_i^4$ & $P_i^5$ & $P_i^6$ & $P_i^7$  & $P_i^8$ \\
\hline 
$a_1$ & $a_1$ & $a_2$ & $a_2$ & $a_3$ & $a_3$ & $a_4$  & $a_4$ \\ 
$a_2$ & $a_4$ & $a_1$ & $a_3$ & $a_2$ & $a_4$ & $a_1$  & $a_3$ \\
$a_3$ & $a_3$ & $a_4$ & $a_4$ & $a_1$ & $a_1$ & $a_2$  & $a_2$ \\ 
$a_4$ & $a_2$ & $a_3$ & $a_1$ & $a_4$ & $a_2$ & $a_3$  & $a_1$ 
\end{tabular}
\caption{Circular Domain} 
\label{T1} 
\end{table}
\end{example}

The following remark is immediate.{\footnote{Throughout the paper, we use $\subset$ and $\subseteq$ to denote strict and weak set inclusion, respectively.}}

\begin{remark}\label{R1}$\mathbb{D}^{LDICT}\subset \mathbb{D}^{DICT}$
\end{remark}

Characterizing the set of domains in which unanimity and strategy-proofness imply dictatorship remains an open problem in the literature: despite a substantial body of work, the structure of $\mathbb{D}^{DICT}$ is not fully understood. By Remark \ref{R1}, $\mathbb{D}^{LDICT}$ is a strict subset of $\mathbb{D}^{DICT}$. Moreover, local strategy-proofness is often the relevant incentive constraint agents face, rather than a mere relaxation of strategy-proofness. We therefore characterize $\mathbb{D}^{LDICT}$ directly.

\section{Results}\label{sec3}

\subsection{A Necessary Condition}

In this section, we provide a necessary condition on domains where locally strategy-proof and unanimous scfs satisfy dictatorship. Let $\mathbb{D}$ be a domain. A path $(P^1_i,P^2_i,\ldots,P^t_i)$ in $\mathbb{D}$ is a sequence of distinct preferences in $\mathbb{D}$ satisfying the property that consecutive preferences are adjacent, i.e., $P^k_i\sim P^{k+1}_i$ for all $k = 1,\ldots,t-1$. We call the path $(P^1_i,P^2_i,\ldots,P^t_i)$ as a path from $P^1_i$ to $P^t_i$ in $\mathbb{D}$. A pair of preference orderings $P_i, P_i'\in \mathbb{D}$ is connected in $\mathbb{D}$ if there exists a path from $P_i$ to $P_i'$ in $\mathbb{D}$. $\mathbb{D}$ is connected if every pair of preference orderings $P_i, P_i'\in \mathbb{D}$ is connected in $\mathbb{D}$. 

Let $\bar{\mathbb{D}}\subseteq \mathbb{D}$ be a sub-domain of $\mathbb{D}$. An ordering $P_i$ is a neighbour of $\bar{\mathbb{D}}$ in $\mathbb{D}$ if $P_i\in \mathbb{D}\setminus \bar{\mathbb{D}}$ and there exists an ordering $P'_i\in \bar{\mathbb{D}}$ such that $P'_i$ and $P_i$ are adjacent. For any $\bar{\mathbb{D}}\subseteq \mathbb{D}$, let $N(\bar{\mathbb{D}}, \mathbb{D})$ be the set of neighbours of $\bar{\mathbb{D}}$ in $\mathbb{D}$, i.e., $N(\bar{\mathbb{D}}, \mathbb{D})=\{P_i\in \mathbb{P}: P_i \textit{ is a neighbour of } \bar{\mathbb{D}} \textit{ in } \mathbb{D}\}$.

A pair $P_i, P_i'\in \mathbb{D}$ is top-connected in $\mathbb{D}$ if there exists a path $(P_i=P^1_i,P^2_i,\ldots,P^t_i=P_i')$ in $\mathbb{D}$ such that $r_1(P_i^1)=r_1(P_i^2)=\ldots =r_1(P_i^t)$.
For any $P_i\in \mathbb{D}$, the top-connected closure of $P_i$ in $\mathbb{D}$, denoted as $\mathbb{D}^{TCC}(P_i)$, is the set of preferences in $\mathbb{D}$ that are top-connected to $P_i$ in $\mathbb{D}$, i.e., $\mathbb{D}^{TCC}(P_i)=\{P'_i\in \mathbb{D}: P_i \text{ and } P_i' \text{ are top-connected in } \mathbb{D}\}\cup P_i$. Note that, for any $P_i\in \mathbb{D}$, $\mathbb{D}^{TCC}(P_i)\subseteq \mathbb{D}$. For any $P_i\in \mathbb{D}$, $\mathbb{D}^{TCC}(P_i)$ has two distinct neighbours in $\mathbb{D}$ if there exist $P'_i, P''_i\in N(\mathbb{D}^{TCC}(P_i),\mathbb{D})$  such that $r_1(P'_i)\neq r_1(P''_i)$.

A domain is connected with distinct neighbours if it satisfies two conditions: $(1)$ the domain is connected (i.e., any two preferences are connected by a sequence of adjacent preferences), and $(2)$ for every preference $P_i$, its top-connected closure ( i.e., the set of all preferences that are top-connected to $P_i$) has at least two neighbouring preferences with distinct top alternatives.

\begin{defn} A domain $\mathbb{D}$ is \textbf{connected with distinct neighbours} if 
\begin{enumerate}
\item \textbf{Connectedness}: $\mathbb{D}$ is connected, and
\item \textbf{Distinct neighbours}: For any $P_i\in \mathbb{D}$, $\mathbb{D}^{TCC}(P_i)$ has two distinct neighbours in $\mathbb{D}$.
\end{enumerate}
\end{defn}

In the following, we will illustrate connected with distinct neighbours domains via examples.

\begin{example}\label{ex2}\rm Let $A=\{a_1,a_2,a_3,a_4\}$, and consider the domain $\mathbb{D}=\{P_i^1,P_i^2,P_i^3,P_i^4,P_i^5,P_i^6,P_i^7\}$,  as specified in Table \ref{T2}. Notably, $\mathbb{D}$ constitutes a single-crossing domain with respect to the linear order $a_1> a_2> a_3> a_4$ and the linear order $P_i^1\succ P_i^2\succ P_i^3\succ P_i^4\succ P_i^5\succ P_i^6\succ P_i^7$ (see \cite{saporiti2009strategy}). 
\begin{table}[h!]
\centering
\begin{tabular}{ c c c c c c c }
$P_i^1$ & $P_i^2$ & $P_i^3$ & $P_i^4$ & $P_i^5$ & $P_i^6$ & $P_i^7$   \\
\hline 
$a_1$ & $a_2$ &$ a_2$ & $a_3$ & $a_3$ & $a_3$ & $a_4$ \\ 
$a_2$ & $a_1$ & $a_3$ & $a_2$ & $a_2$ & $a_4$ & $a_3$ \\
$a_3$ & $a_3$ &$ a_1$ & $a_1$ & $a_4$ & $a_2$ & $a_2$ \\ 
$a_4$ & $a_4$ & $a_4$ & $a_4$ & $a_1$ & $a_1$ & $a_1$ 
\end{tabular}
\caption{Single crossing domain} 
\label{T2} 
\end{table}
An important property of $\mathbb{D}$ is that $P_i^2$ is the single ordering that is adjacent to $P_i^1$ in $\mathbb{D}$. Similarly, $P_i^6$ is the single ordering that is adjacent to $P_i^7$ in $\mathbb{D}$. For any other ordering $P_i^j$, $j\in\{2,3,4,5,6\}$, we have that only $P_i^{j-1}$ and $P_i^{j+1}$ are adjacent to $P_i^j$ in $\mathbb{D}$ (see Figure \ref{fig:ex2}). 

From the observation in the previous paragraph, $\mathbb{D}$ is a connected domain. Moreover, $\mathbb{D}^{TCC}(P_i^1)=\{P_i^1\}$, $\mathbb{D}^{TCC}(P_i^2)=\mathbb{D}^{TCC}(P_i^3)=\{P_i^2,P_i^3\}$, $\mathbb{D}^{TCC}(P_i^4)=\mathbb{D}^{TCC}(P_i^5)=\mathbb{D}^{TCC}(P_i^6)=\{P_i^4, P_i^5, P_i^6\}$ and $\mathbb{D}^{TCC}(P_i^7)=\{P_i^7\}$. Note that,  $\mathbb{D}$ does not satisfy the distinct neighbours property. For instance, $N(\mathbb{D}^{TCC}(P_i^1),\mathbb{D})=\{P_i^2\}$. \hfill $\blacksquare$
\end{example}

\begin{figure}[h!]
\centering
\includegraphics[width=0.7\textwidth]{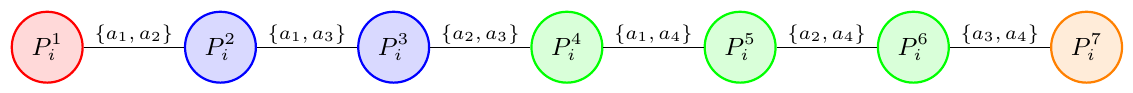}
\caption{Paths induced by the adjacent preferences in the domain in Example 2}
\label{fig:ex2}
\end{figure}

\begin{example}\label{ex3}\rm Let $A=\{a_1,a_2,a_3,a_4\}$, and consider the domain $$\mathbb{D}=\{P_i^1,P_i^2,P_i^3,P_i^4,P_i^5,P_i^6,P_i^7,P_i^8\}$$ as specified in Table \ref{T3}. Note that, $\mathbb{D}$ constitutes a single-peaked  domain with respect to the linear order $a_1> a_2> a_3> a_4$ (see \cite{moulin1980strategy}). 
\begin{table}[h!]
\centering
\begin{tabular}{ c c c c c c c c}
$P_i^1$ & $P_i^2$ & $P_i^3$ & $P_i^4$ & $P_i^5$ & $P_i^6$ & $P_i^7$ & $P_i^8$   \\
\hline 
$a_1$ & $a_2$ & $a_2$ & $a_2$ & $a_3$ & $a_3$ & $a_3$ & $a_4$ \\ 
$a_2$ & $a_1$ & $a_3$ & $a_3$ & $a_2$ & $a_2$ & $a_4$ & $a_3$ \\
$a_3$ & $a_3$ &$ a_1$ & $a_4$ & $a_4$ & $a_1$ & $a_2$ & $a_2$ \\ 
$a_4$ & $a_4$ & $a_4$ & $a_1$ & $a_1$ & $a_4$ & $a_1$ & $a_1$ 
\end{tabular}
\caption{Single-peaked domain} 
\label{T3} 
\end{table}
 
Observe that $\mathbb{D}$ is a connected domain (see Figure \ref{fig:ex3}). Furthermore, $\mathbb{D}^{TCC}(P_i^1)=\{P_i^1\}$, $\mathbb{D}^{TCC}(P_i^2)=\mathbb{D}^{TCC}(P_i^3)=\mathbb{D}^{TCC}(P_i^4)=\{P_i^2,P_i^3,P_i^4\}$, $\mathbb{D}^{TCC}(P_i^5)=\mathbb{D}^{TCC}(P_i^6)=\mathbb{D}^{TCC}(P_i^7)=\{P_i^5, P_i^6, P_i^7\}$ and $\mathbb{D}^{TCC}(P_i^8)=\{P_i^8\}$. Notably,  $\mathbb{D}$ does not satisfy the distinct neighbours property.  For example, $N(\mathbb{D}^{TCC}(P_i^1),\mathbb{D})=\{P_i^2\}$.
 \hfill $\blacksquare$
\end{example}

\begin{figure}[h!]
\centering
\includegraphics[width=0.7\textwidth]{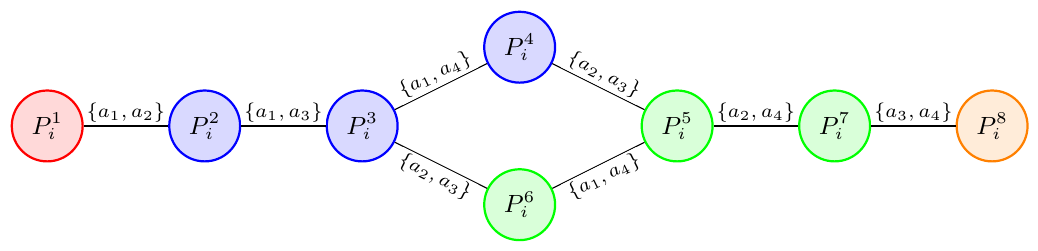}
\caption{Paths induced by the adjacent preferences in the domain in Example 3}
\label{fig:ex3}
\end{figure}

\begin{example}\label{ex4}\rm Let $A=\{a_1,a_2,a_3,a_4\}$, and consider the domain $\mathbb{D}$ as specified in Table \ref{T4}: $$\mathbb{D}=\{P_i^1,P_i^2,P_i^3,P_i^4,P_i^5,P_i^6,P_i^7,P_i^8,P_i^9,P_i^{10},
P_i^{11},P_i^{12}\}.$$ 
\begin{table}[h!]
\centering
\begin{tabular}{ c c c c c c c c c c c c}
$P_i^1$ & $P_i^2$ & $P_i^3$ & $P_i^4$ & $P_i^5$ & $P_i^6$ & $P_i^7$ & $P_i^8$ & $P_i^9$ & $P_i^{10}$ & $P_i^{11}$ & $P_i^{12}$  \\
\hline 
$a_1$ & $a_2$ & $a_2$ & $a_3$ & $a_3$ & $a_3$ & $a_4$ & $a_4$ & $a_4$ & $a_4$ & $a_1$ & $a_1$ \\ 
$a_2$ & $a_1$ & $a_3$ & $a_2$ & $a_2$ & $a_4$ & $a_3$ & $a_3$ & $a_1$ & $a_1$ & $a_4$ & $a_2$ \\
$a_3$ & $a_3$ & $a_1$ & $a_1$ & $a_4$ & $a_2$ & $a_2$ & $a_1$ & $a_3$ & $a_2$ & $a_2$ & $a_4$ \\ 
$a_4$ & $a_4$ & $a_4$ & $a_4$ & $a_1$ & $a_1$ & $a_1$ & $a_2$ & $a_2$ & $a_3$ & $a_3$ & $a_3$ 
\end{tabular}
\caption{Union of single-peaked and single-dipped domain} 
\label{T4} 
\end{table}
Note that $\mathbb{D}$ is the union of a single-peaked domain
$\{P_i^1,\ldots,P_i^7\}$ and a single-dipped domain
$\{P_i^1,P_i^7,P_i^8,\ldots,P_i^{12}\}$, both with respect to the linear order $a_1> a_2> a_3> a_4$ (this example is taken from Example 5 in \cite{hong2023unanimity}). 

Observe that $\mathbb{D}$ is a connected domain (see Figure \ref{fig:ex4}). Moreover, $\mathbb{D}^{TCC}(P_i^1)=\mathbb{D}^{TCC}(P_i^{11})=\mathbb{D}^{TCC}(P_i^{12})=\{P_i^1,P_i^{11},P_i^{12}\}$, $\mathbb{D}^{TCC}(P_i^2)=\mathbb{D}^{TCC}(P_i^3)=\{P_i^2,P_i^3\}$, $\mathbb{D}^{TCC}(P_i^4)=\mathbb{D}^{TCC}(P_i^5)=\mathbb{D}^{TCC}(P_i^6)=\{P_i^4,P_i^5,P_i^6\}$ and $\mathbb{D}^{TCC}(P_i^7)=\mathbb{D}^{TCC}(P_i^8)=\mathbb{D}^{TCC}(P_i^9)=\mathbb{D}^{TCC}(P_i^{10})=\{P_i^7,P_i^8,P_i^9,P_i^{10}\}$. It can be verified that $\mathbb{D}$ satisfies the distinct neighbours property.  \hfill $\blacksquare$
\end{example}

\begin{figure}[h!]
\centering
\includegraphics[width=0.7\textwidth]{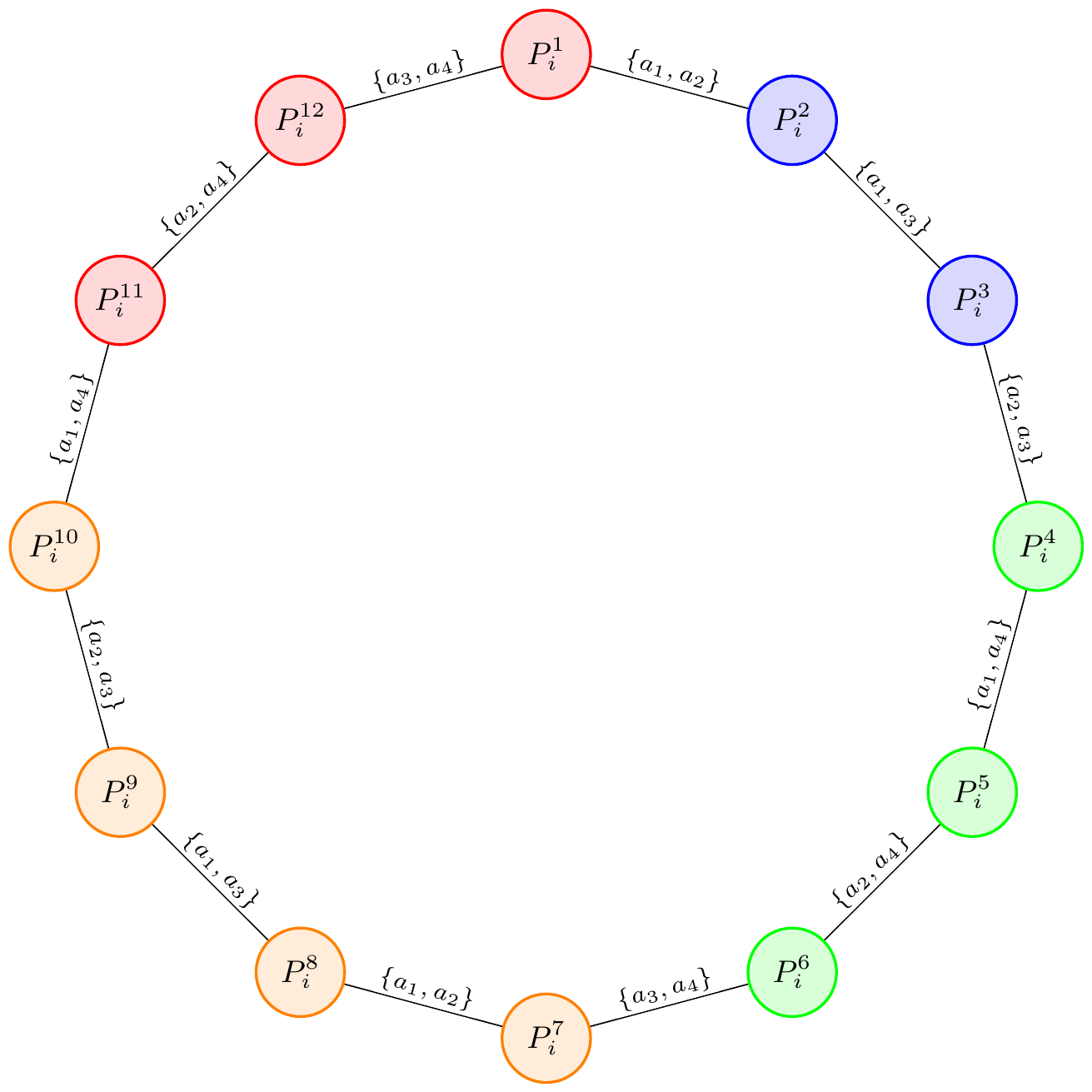}
\caption{Paths induced by the adjacent preferences in the domain in Example 4}
\label{fig:ex4}
\end{figure}

Note that a domain being connected with distinct neighbours does not force preferences sharing a common top-ranked alternative to be top-connected to one another. The following example illustrates this.
 
\begin{example}\label{ex5}\rm Let $A=\{a_1,a_2,a_3,a_4,a_5\}$, and consider the domain $\mathbb{D}=\{P_i^1,\ldots,P_i^{21}\}$ as specified in Table \ref{T5}.
 
\begin{table}[h!]
\centering\tiny
\setlength{\tabcolsep}{5pt}
\begin{tabular}{ c c c c c c c c c c c c c c c c c c c c c }
$P_i^1$ & $P_i^2$ & $P_i^3$ & $P_i^4$ & $P_i^5$ & $P_i^6$ & $P_i^7$ & $P_i^8$ & $P_i^9$ & $P_i^{10}$ & $P_i^{11}$ & $P_i^{12}$ & $P_i^{13}$ & $P_i^{14}$ & $P_i^{15}$ & $P_i^{16}$ & $P_i^{17}$ & $P_i^{18}$ & $P_i^{19}$ & $P_i^{20}$ & $P_i^{21}$ \\
\hline
$a_1$ & $a_1$ & $a_1$ & $a_1$ & $a_1$ & $a_1$ & $a_1$ & $a_1$ & $a_1$ & $a_2$ & $a_2$ & $a_2$ & $a_2$ & $a_3$ & $a_3$ & $a_3$ & $a_3$ & $a_4$ & $a_4$ & $a_5$ & $a_5$ \\
$a_2$ & $a_2$ & $a_3$ & $a_3$ & $a_4$ & $a_4$ & $a_4$ & $a_4$ & $a_5$ & $a_1$ & $a_1$ & $a_1$ & $a_3$ & $a_1$ & $a_1$ & $a_1$ & $a_2$ & $a_1$ & $a_5$ & $a_1$ & $a_4$ \\
$a_3$ & $a_4$ & $a_2$ & $a_4$ & $a_2$ & $a_3$ & $a_5$ & $a_5$ & $a_4$ & $a_3$ & $a_4$ & $a_4$ & $a_1$ & $a_2$ & $a_4$ & $a_4$ & $a_1$ & $a_5$ & $a_1$ & $a_4$ & $a_1$ \\
$a_4$ & $a_5$ & $a_4$ & $a_5$ & $a_5$ & $a_5$ & $a_2$ & $a_3$ & $a_2$ & $a_4$ & $a_3$ & $a_5$ & $a_4$ & $a_4$ & $a_2$ & $a_5$ & $a_4$ & $a_2$ & $a_2$ & $a_2$ & $a_2$ \\
$a_5$ & $a_3$ & $a_5$ & $a_2$ & $a_3$ & $a_2$ & $a_3$ & $a_2$ & $a_3$ & $a_5$ & $a_5$ & $a_3$ & $a_5$ & $a_5$ & $a_5$ & $a_2$ & $a_5$ & $a_3$ & $a_3$ & $a_3$ & $a_3$
\end{tabular}
\caption{Preferences sharing the same top need not be top-connected}
\label{T5}
\end{table}
 
Observe that $\mathbb{D}$ is a connected domain (see Figure \ref{fig:ex5}). Moreover,
$\mathbb{D}^{TCC}(P_i^1)=\mathbb{D}^{TCC}(P_i^3)=\{P_i^1,P_i^3\}$, and $\mathbb{D}^{TCC}(P_i^k)= \{P_i^2,P_i^4,P_i^5,P_i^6,P_i^7,P_i^8,P_i^9\}$ for every
$k\in\{2,4,5,6,7,8,9\}$. Similarly, $\mathbb{D}^{TCC}(P_i^k)=\{P_i^{10},P_i^{11},P_i^{12},P_i^{13}\}$ for every $k\in\{10,11,12,13\}$, and $\mathbb{D}^{TCC}(P_i^k)=\{P_i^{14},P_i^{15},P_i^{16},P_i^{17}\}$ for every $k\in\{14,15,16,17\}$. Finally, $\mathbb{D}^{TCC}(P_i^{18})=\mathbb{D}^{TCC}(P_i^{19})=\{P_i^{18},P_i^{19}\}$
and $\mathbb{D}^{TCC}(P_i^{20})=\mathbb{D}^{TCC}(P_i^{21})=\{P_i^{20},P_i^{21}\}$.
 
It can be verified that $\mathbb{D}$ satisfies the distinct neighbours property, and hence $\mathbb{D}$ is connected with distinct neighbours. However, two preferences sharing the same top need not be top-connected to each other. For instance, $r_1(P_i^1)=r_1(P_i^2)=a_1$, but $P_i^2\notin\mathbb{D}^{TCC}(P_i^1)$, i.e., $P_i^1$ and $P_i^2$ are not top-connected in $\mathbb{D}$. \hfill $\blacksquare$
\end{example}

\begin{figure}[h!]
\centering
\includegraphics[width=0.7\textwidth]{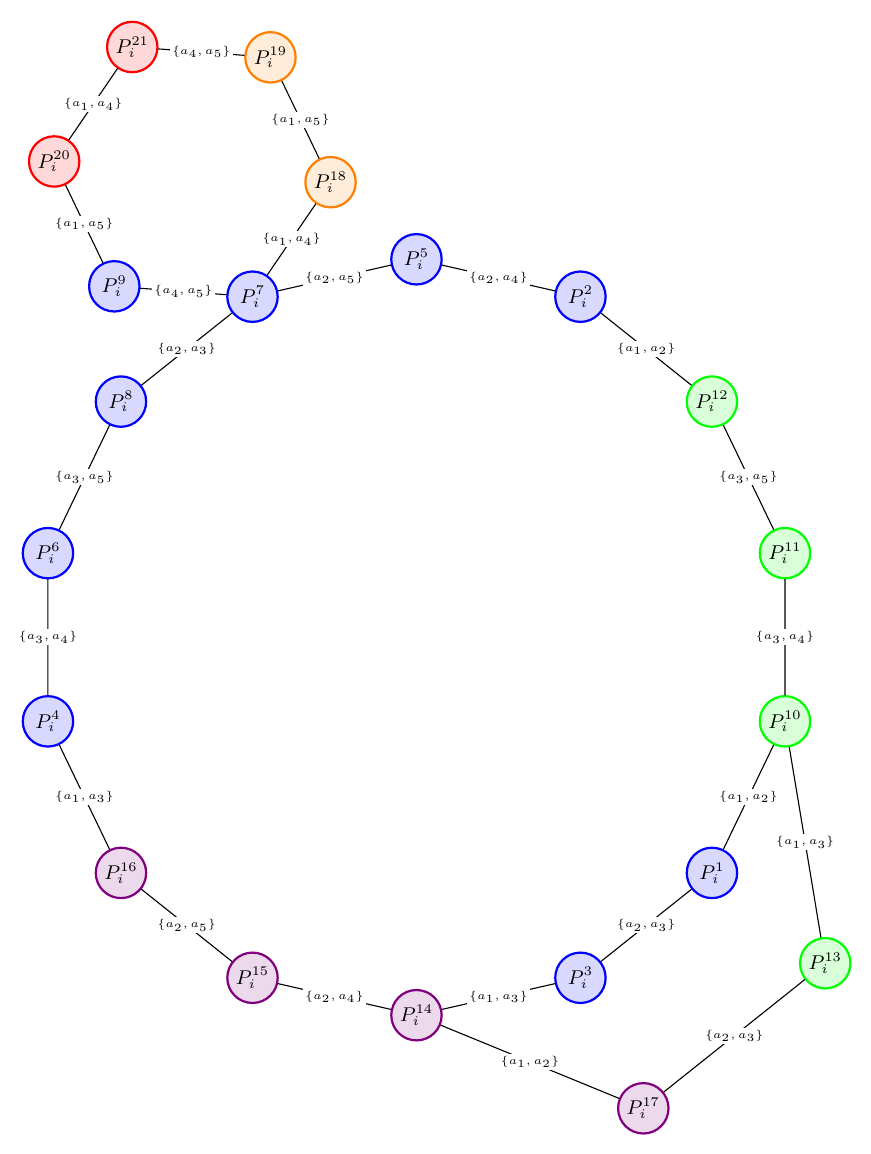}
\caption{Paths induced by the adjacent preferences in the domain in Example 5}
\label{fig:ex5}
\end{figure}

We show that domains where local strategy-proofness and unanimity imply dictatorship are connected with distinct neighbours.

\begin{theorem}\label{thm:1}
Let $\mathbb{D}$ be a minimally rich domain. If any unanimous and locally strategy-proof scf $f: \mathbb{D}^n\rightarrow A$ satisfies dictatorship, then $\mathbb{D}$ is connected with distinct neighbours. 
\end{theorem}

The proof of theorem \ref{thm:1} is in the Appendix. The proof follows a contrapositive argument. Suppose the domain $\mathbb{D}$ is not connected with distinct neighbours. We then construct a unanimous and locally strategy-proof scf that does not satisfy dictatorship. In particular, for each case where $\mathbb{D}$ violates the connected with distinct neighbours property, we exhibit an scf satisfying all axioms but avoiding dictatorship. 

We relate our necessary condition to a necessary condition on domains where strategy-proofness and unanimity imply dictatorship established by \cite{Aswal03}.
A domain satisfies the \textit{unique seconds property} if there exist two alternatives $a$ and $b$ such that whenever $a$ is ranked first in a preference in the domain, $b$ is ranked second.
\cite{Aswal03} show that if strategy-proofness and unanimity imply dictatorship on a minimally rich domain, then the domain must violate the unique seconds property. It is easy to verify that minimally rich domains satisfying the connected with distinct neighbours property violate the unique seconds property. Therefore, the necessary condition identified in our paper is stronger than that of \cite{Aswal03}.
In particular, the implication is strict. For instance, the domain in Example~\ref{ex1} violates the unique seconds property and hence satisfies the necessary condition of \cite{Aswal03}, but it does not satisfy the connected with distinct neighbours property. Note that this domain is not connected. However, connectedness alone does not close the gap: the domain $\bar{\mathbb{D}}$ in Example~\ref{ex7} is connected and violates the unique seconds property, but it does not satisfy the distinct neighbours property, since the top-connected closures of $P_i^1$ and $P_i^{10}$ each have only one neighbour. Thus, even if a domain is connected, violation of the unique seconds property does not imply the distinct neighbours property.

Let $\mathbb{D}^{CDN}$ be the set of minimally rich domains that are connected with distinct neighbours. The following remark follows directly from Theorem \ref{thm:1}.
\begin{remark}\label{R2} $\mathbb{D}^{LDICT}\subseteq \mathbb{D}^{CDN}$.
\end{remark}

This raises the natural question of whether the two classes in fact coincide: does $\mathbb{D}^{LDICT}=\mathbb{D}^{CDN}$?  Settling this in general turns out to be genuinely hard, and we leave it open. What we can show is a precise answer under one additional restriction: the next section identifies exactly when equality holds.

\subsection{A Sufficient Condition} 
In this section, we investigate whether the necessary condition in Theorem \ref{thm:1} is also sufficient. A key difficulty is that we are unable to establish that, on domains satisfying connected with distinct neighbours property, local strategy-proofness and unanimity imply tops-onlyness. However, note that any dictatorial scf necessarily satisfies tops-onlyness. This observation motivates restricting attention to domains where local strategy-proofness and unanimity imply tops-onlyness, which we refer to as L-tops-only domains.\footnote{Domains on which incentive and efficiency conditions imply tops-onlyness have been studied in the literature; see, for instance, \cite{Weymark08}, \cite{Chatt11} and \cite{chatterji2018random}. Thus, this restriction is not ad hoc.} Within this class, we show that connected with distinct neighbours property is sufficient for dictatorship. A tops-only scf selects an alternative based solely on the top-ranked alternatives in each individual's preference ordering, disregarding how the remaining alternatives are ranked.

\begin{defn} A scf $f: \mathbb{D}^n\rightarrow A$ satisfies tops-onlyness if for any $P,P'\in \mathbb{D}^n$ such that $r_1(P_i)=r_1(P'_i)$ for all $i\in N$, we have $f(P)=f(P')$. 
\end{defn}

Below we provide a formal definition of L-tops-only domains. Note that restricting attention to these domains is not vacuous and is satisfied by several domains studied in the literature (see \cite{kumar2021localb} and \cite{hong2023unanimity} for examples).

\begin{defn} A domain $\mathbb{D}\subseteq \mathbb{P}$ is L-tops-only if any unanimous and local strategy-proof scf $f: \mathbb{D}^n\rightarrow A$ satisfies tops-onlyness.
\end{defn}

We can now state the main result of this section.

\begin{theorem}\label{thm:2}
Let $\mathbb{D}$ be a minimally rich and L-tops-only domain. If $\mathbb{D}$ is connected with distinct neighbours, then any unanimous and locally strategy-proof scf $f: \mathbb{D}^n\rightarrow A$ satisfies dictatorship.
\end{theorem}

The proof of Theorem \ref{thm:2} is given in the Appendix. Since
$\mathbb{D}$ is L-tops-only, every unanimous and locally strategy-proof scf on $\mathbb{D}$ is automatically tops-only, so the substantive work reduces to showing that unanimity, tops-onlyness, and local strategy-proofness together force dictatorship on a minimally rich domain that is connected with distinct neighbours. We associate with $\mathbb{D}$ an undirected graph $G(\mathbb{D})$ whose vertices are the alternatives, with an edge between two alternatives whenever some pair of adjacent preferences in $\mathbb{D}$ has those alternatives as their respective tops. We first show that $G(\mathbb{D})$ is connected and every vertex has degree at least two, and hence contains a cycle (Lemma \ref{L1}). The proof then proceeds by induction on the number of voters. For two voters, we establish two main facts using the structure of
$G(\mathbb{D})$. First, for every edge $(a,b)$ of $G(\mathbb{D})$,
either voter 1 is decisive over $a$ or voter 2 is decisive over $b$ (Lemma \ref{L3}). Second, since every vertex has degree at least two, for every alternative $a$, either voter 1 or voter 2 is decisive over $a$ (Lemma \ref{L4}). Applying the second fact to every alternative, one of the two voters must be decisive over every alternative, since $f$ being a function rules out one voter being decisive over one alternative while the other is decisive over a different one. Therefore, one of the two voters must be a dictator. For $n>2$ voters, we use a cloning argument, reducing to the $n-1$-voter case by having two voters report identically, applying the induction hypothesis to the resulting scf, and then extending the resulting dictator back to the original $n$-voter scf using the same decisiveness argument.

Let $\mathbb{D}^{L-TOPS-ONLY}$ denote the set of minimally rich domains that are L-tops-only. Since every dictatorial scf is tops-only, we have $\mathbb{D}^{LDICT}\subseteq \mathbb{D}^{L-TOPS-ONLY}$. This inclusion is strict: the domains in Examples \ref{ex2} and \ref{ex3} belong to $\mathbb{D}^{L-TOPS-ONLY}$ but are not connected with distinct neighbours, so by Remark \ref{R2} they do not belong to $\mathbb{D}^{LDICT}$.

\begin{remark}\label{R3}  $\mathbb{D}^{LDICT}\subset \mathbb{D}^{L-TOPS-ONLY}$.
\end{remark}

Theorem \ref{thm:2} gives $\mathbb{D}^{CDN}\cap \mathbb{D}^{L-TOPS-ONLY}\subseteq \mathbb{D}^{LDICT}$, while Remarks \ref{R2} and \ref{R3} give $\mathbb{D}^{LDICT}\subseteq \mathbb{D}^{CDN}\cap \mathbb{D}^{L-TOPS-ONLY}$. Together, these yield the following characterization.

\begin{remark}\label{R4}  $\mathbb{D}^{LDICT}=\mathbb{D}^{CDN}\cap \mathbb{D}^{L-TOPS-ONLY}$.
\end{remark}

If every connected with distinct neighbours domain were L-tops-only, that is, if $\mathbb{D}^{CDN}\subseteq \mathbb{D}^{L-TOPS-ONLY}$, then Remark \ref{R4} would give $\mathbb{D}^{LDICT}=\mathbb{D}^{CDN}$, a complete characterization. Theorem \ref{thm:2} does not establish this inclusion. We note, however, that several domains satisfy both properties. The domains in Examples \ref{ex2}--\ref{ex5} and the unrestricted domain $\mathbb{P}$ all belong to $\mathbb{D}^{L-TOPS-ONLY}$.\footnote{By Corollary 2 of \cite{kumar2021localb}, every unanimous and locally strategy-proof scf on a domain satisfying the pairwise no-restoration property of \cite{sato2013sufficient} (Property P, formally defined in Section \ref{sec5}) is tops-only. All the domains listed here satisfy Property P. For Examples \ref{ex2}--\ref{ex4} the verification is straightforward, while for Example \ref{ex5} it is lengthier, and we leave it to the reader.} The domains in Examples \ref{ex2} and \ref{ex3} are not connected with distinct neighbours, and hence do not belong to $\mathbb{D}^{LDICT}$. In contrast, the domains in Examples \ref{ex4} and \ref{ex5}, as well as the unrestricted domain $\mathbb{P}$, are also connected with distinct neighbours, and hence, by Remark \ref{R4}, belong to $\mathbb{D}^{LDICT}$.

\section{Strengthening the Sufficiency Result Beyond L-tops-only Domains}\label{sec4}

In the previous section, we established a sufficiency result under the additional assumption that the domain is L-tops-only. A natural question is whether this restriction can be dispensed with, i.e., whether the connected with distinct neighbours property alone ensures that unanimity and local strategy-proofness imply dictatorship. It remains open whether every connected with distinct neighbours domain is L-tops-only. If some are not, Theorem \ref{thm:2} is silent on them, which raises a concern about the scope of our sufficiency result.

We address this concern by strengthening the incentive requirement. Recall that a domain is a \textit{dictatorial domain} if every unanimous and strategy-proof scf defined on it is dictatorial (\cite{Aswal03}). Theorem \ref{thm:3} shows that every connected with distinct neighbours domain, L-tops-only or not, is a dictatorial domain. In exchange for dropping L-tops-only, Theorem \ref{thm:3} requires the stronger incentive property of strategy-proofness in place of local strategy-proofness.

\begin{theorem}\label{thm:3}
Let $\mathbb{D}$ be a minimally rich and connected with distinct neighbours domain. Then any unanimous and strategy-proof scf $f: \mathbb{D}^n\rightarrow A$ satisfies dictatorship.
\end{theorem}

The proof of Theorem \ref{thm:3} is provided in the Appendix. By a result of \cite{Aswal03}, a minimally rich domain on which every unanimous and strategy-proof scf is dictatorial for two voters has the same property for any number of voters (see Proposition \ref{P3}), so it suffices to consider two voters. For two voters, the proof follows the same steps as that of Theorem \ref{thm:2}, with one critical difference. There, tops-onlyness is guaranteed by the L-tops-only hypothesis, whereas here we show directly that unanimity and strategy-proofness imply tops-onlyness on connected with distinct neighbours domains.

Theorem \ref{thm:3} establishes that $\mathbb{D}^{CDN}\subseteq \mathbb{D}^{DICT}$. This inclusion is strict, since the domain in Example \ref{ex1} is dictatorial but is not connected with distinct neighbours.

\begin{remark}\label{R5} $\mathbb{D}^{CDN}\subset \mathbb{D}^{DICT}$.
\end{remark}

Consequently, every domain $\mathbb{D}\in \mathbb{D}^{CDN}$ falls into one of two cases. If $\mathbb{D}\in \mathbb{D}^{L-TOPS-ONLY}$, then Theorem \ref{thm:2} gives $\mathbb{D}\in \mathbb{D}^{LDICT}$. If $\mathbb{D}\notin \mathbb{D}^{L-TOPS-ONLY}$, then Theorem \ref{thm:3} still gives $\mathbb{D}\in \mathbb{D}^{DICT}$.

Theorem \ref{thm:3} also bears on two strands of the literature. The first characterizes dictatorial domains, i.e., the set $\mathbb{D}^{DICT}$. Several sufficient conditions have been established (see, for instance, \cite{Aswal03}, \cite{sato2010circular}, \cite{pramanik2015further} and \cite{chatterji2023taxonomy}), but a complete characterization remains an open question. Theorem \ref{thm:3} adds a new sufficient condition. The second strand identifies domains on which unanimity and strategy-proofness imply tops-onlyness, commonly called tops-only domains. Sufficient conditions have been proposed (see, for instance, \cite{Weymark08}, \cite{Chatt11} and \cite{chatterji2018random}), but a full characterization has yet to be achieved. Since dictatorship implies tops-onlyness, Theorem \ref{thm:3} shows that every connected with distinct neighbours domain is a tops-only domain.

\section{Discussion}\label{sec5}

In this section, we discuss the connection between our results and existing findings in the literature, and demonstrate how it offers a new perspective on these results.

\subsection{Results on local-global equivalence for unanimous scfs}
Recent literature studies preference domains where  local strategy-proofness implies strategy-proofness for unanimous scfs (see, for instance, \cite{kumar2021localb} and \cite{hong2023unanimity}). This property is referred to as the local-global equivalence for unanimous scfs (uLGE property).

The pairwise no-restoration property (Property P) was introduced by \cite{sato2013sufficient}, and \cite{kumar2021localb} show that it is sufficient for uLGE.\footnote{\cite{kumar2021localb} establish this for random scfs. The result for deterministic scfs is their Corollary 1.} We first define Property P and then relate it to the connected with distinct neighbours property. Let $\mathbb{D}$ be a domain and $a,b\in A$ be a pair of alternatives. A path $\pi=(P_i^1,\ldots,P_i^t)$ in $\mathbb{D}$ satisfies \textit{ no $\{a,b\}$-restoration} if the relative ranking of $a$ and $b$ is reversed at most once along $\pi$, i.e., there do not exist integers $q$, $r$ and $s$ with $1\leq q < r < s \leq t$ such that either $(i)$ $a\;P_i^q\;b$, $b\;P_i^r\;a$ and $a\;P_i^s\;b$, or $(ii)$ $b\;P_i^q\;a$, $a\;P_i^r\;b$ and $b\;P_i^s\;a$. The domain $\mathbb{D}$ satisfies Property P if for all distinct $P_i,P'_i\in \mathbb{D}$ and distinct $a,b \in A$, there exists a path from $P_i$ to $P'_i$ in $\mathbb{D}$ with no $\{a,b\}$-restoration.

The connected with distinct neighbours property and Property P are independent. The domains in Examples \ref{ex2} and \ref{ex3} satisfy Property P but are not connected with distinct neighbours. Conversely, the domain in Example \ref{ex6} is connected with distinct neighbours but does not satisfy Property P.

\begin{example}\label{ex6}\rm Let $A=\{a_1,a_2,a_3,a_4\}$. Consider the domain $$\mathbb{D}^*=\{P_i^1,P_i^2,P_i^3,P_i^4,P_i^5,P_i^6,P_i^7,P_i^8,P_i^9,
P_i^{10},P_i^{11},P_i^{12},P_i^{13}\}$$  specified in Table \ref{T6}. The domain $\mathbb{D}^*$ is obtained from the domain in Example \ref{ex4} by adding a single preference, $P_i^{13}$; the preferences $P_i^1,\ldots,P_i^{12}$ are the same in both examples. Figure \ref{fig:ex6} shows all paths induced by the adjacent preferences in $\mathbb{D}^*$. One can verify that $\mathbb{D}^*$ satisfies the connected with distinct neighbours property. We omit the detailed argument, as the verification is straightforward. However, the domain $\mathbb{D}^*$ does not satisfy Property P. As highlighted in Figure \ref{fig:ex6}, there are exactly two paths between $P_i^{13}$ and $P_i^4$. One path is $(P_i^{13},P_i^1,P_i^2,P_i^3,P_i^4)$, and another is $(P_i^{13},P_i^1,P_i^{12},P_i^{11},P_i^{10},P_i^9,P_i^8,
P_i^7,P_i^6,P_i^5,P_i^4)$. It can be verified that for the pair $a_2$ and $a_3$, neither of these paths satisfies no $\{a_2,a_3\}$-restoration. Therefore, the domain $\mathbb{D}^*$ does not satisfy Property P. \hfill $\blacksquare$

\begin{table}[h!]
\centering
\begin{tabular}{ c c c c c c c c c c c c c c}
$P_i^1$ & $P_i^2$ & $P_i^3$ & $P_i^4$ & $P_i^5$ & $P_i^6$ & $P_i^7$ & $P_i^8$ & $P_i^9$ & $P_i^{10}$ & $P_i^{11}$ & $P_i^{12}$ & $P_i^{13}$  \\
\hline 
$a_1$ & $a_2$ & $a_2$ & $a_3$ & $a_3$ & $a_3$ & $a_4$ & $a_4$ & $a_4$ & $a_4$ & $a_1$ & $a_1$ & $a_1$ \\ 
$a_2$ & $a_1$ & $a_3$ & $a_2$ & $a_2$ & $a_4$ & $a_3$ & $a_3$ & $a_1$ & $a_1$ & $a_4$ & $a_2$ & $a_3$ \\
$a_3$ & $a_3$ &$ a_1$ & $a_1$ & $a_4$ & $a_2$ & $a_2$ & $a_1$ & $a_3$ & $a_2$ & $a_2$ & $a_4$ & $a_2$ \\ 
$a_4$ & $a_4$ & $a_4$ & $a_4$ & $a_1$ & $a_1$ & $a_1$ & $a_2$ & $a_2$ & $a_3$ & $a_3$ & $a_3$ & $a_4$ 
\end{tabular}
\caption{The domain $\mathbb{D}^*$} 
\label{T6} 
\end{table}
\end{example}

\begin{figure}[h!]
\centering
\includegraphics[width=0.7\textwidth]{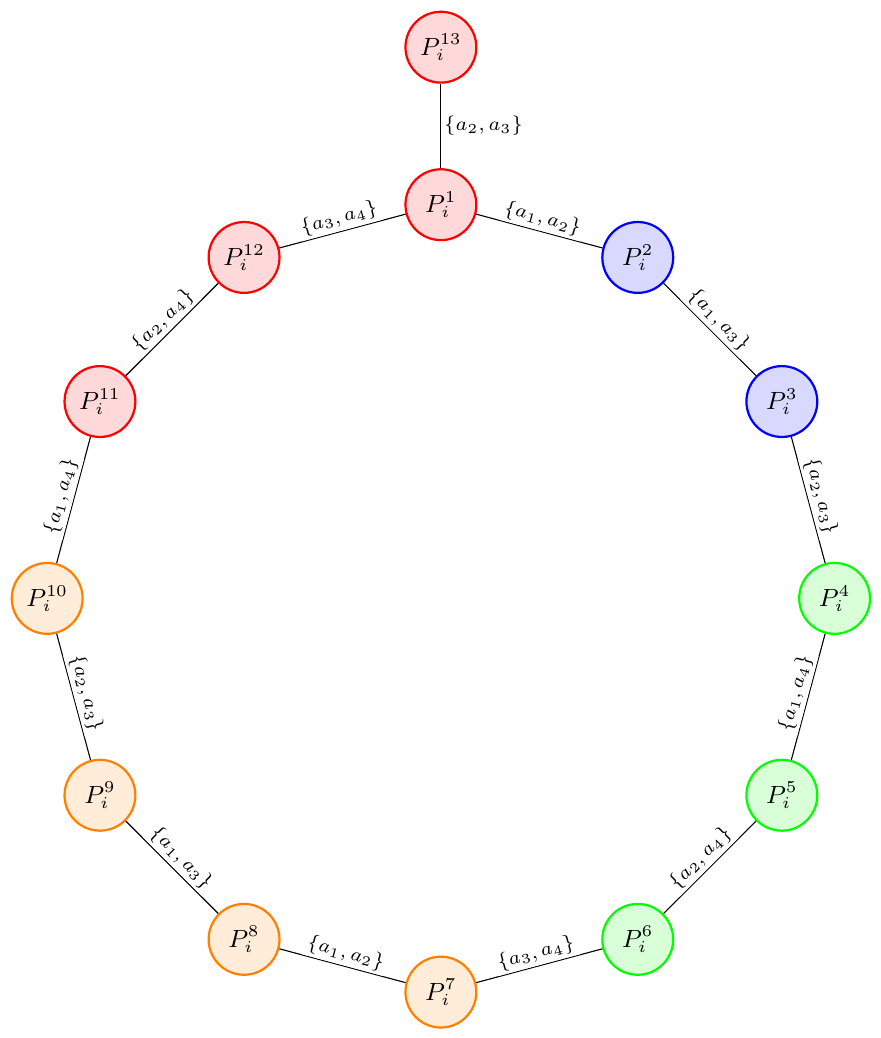}
\caption{Paths induced by the adjacent preferences in $\mathbb{D}^*$}
\label{fig:ex6}
\end{figure}

\cite{hong2023unanimity} study the local version of ordinal Bayesian incentive compatibility for deterministic scfs and establish uLGE for \textit{sparsely connected domains without restoration} (SCD). SCD requires paths without restoration only for pairs of alternatives in which at least one alternative is first-ranked in some preference in the domain, so it is slightly weaker than Property P. On a minimally rich domain, however, every alternative is first-ranked in some preference, and SCD and Property P coincide.

We conclude this subsection by noting that our results contribute new insights to the literature on uLGE. In particular, every domain in $\mathbb{D}^{LDICT}$ satisfies uLGE, since every unanimous and locally strategy-proof scf on it is dictatorial, and hence strategy-proof. For instance, the domain $\mathbb{D}^*$ of Example \ref{ex6} does not satisfy Property P and, being minimally rich, is therefore not SCD. Hence neither the results of \cite{kumar2021localb} nor those of \cite{hong2023unanimity} allow us to conclude that $\mathbb{D}^*$ satisfies the uLGE property. However, $\mathbb{D}^*$ is connected with distinct neighbours, and Remark \ref{R6} below shows that it belongs to $\mathbb{D}^{LDICT}$. Hence $\mathbb{D}^*$ satisfies uLGE. The proof of Remark \ref{R6} is provided in the Appendix.

\begin{remark} \label{R6} The domain $\mathbb{D}^*$ in Example \ref{ex6} belongs to $\mathbb{D}^{LDICT}$, and hence to $\mathbb{D}^{L-TOPS-ONLY}$.
\end{remark}

\subsection{Existing Dictatorship Results under Local Strategy-proofness}
\cite{hong2023unanimity} characterize, within the class of SCD domains, the domains on which unanimity and local strategy-proofness imply dictatorship. They identify the \textit{disagreement property}, which is necessary and sufficient for this implication on SCD domains (see Theorem 1 in \cite{hong2023unanimity}). A domain $\mathbb{D}$ satisfies the disagreement property if for any pair of alternatives $a$ and $b$, such that $r_1(P_i)=a$ and $r_1(P'_i)=b$ for some adjacent preferences $P_i,P'_i\in \mathbb{D}$, there exist preferences $\bar{P_i},\hat{P_i}\in \mathbb{D}$ such that $r_1(\bar{P_i}), r_1(\hat{P_i})\notin \{a,b\}$, $a\;\bar{P_i}\;b$ and $b\;\hat{P_i}\;a$. However, SCD together with the disagreement property is not necessary for unanimity and local strategy-proofness to imply dictatorship. A weaker necessary condition holds in general: if unanimity and local strategy-proofness imply dictatorship on a domain, then the domain is connected and satisfies the disagreement property (see Proposition 4 in \cite{hong2023unanimity}). The following remark shows that our necessary condition is stronger.

\begin{remark} \label{R7}
Let $\mathbb{D}$ be a domain that is connected with distinct neighbours. Then $\mathbb{D}$ is connected and satisfies the disagreement property.
\end{remark}

The proof of Remark \ref{R7} is provided in the Appendix. The implication is strict: the domain in Example \ref{ex7} is connected and satisfies the disagreement property, but it is not connected with distinct neighbours. Hence our necessary condition is strictly stronger than that of \cite{hong2023unanimity}.

\begin{example}\label{ex7}\rm Let $A=\{a_1,a_2,a_3,a_4\}$, and consider the domain $$\bar{\mathbb{D}}=\{P_i^1,P_i^2,P_i^3,P_i^4,P_i^5,P_i^6,P_i^7,P_i^8,P_i^9,P_i^{10}\},$$  as specified in Table \ref{T7}.
 
\begin{table}[h!]
\centering
\begin{tabular}{ c c c c c c c c c c }
$P_i^1$ & $P_i^2$ & $P_i^3$ & $P_i^4$ & $P_i^5$ & $P_i^6$ & $P_i^7$ & $P_i^8$ & $P_i^9$ & $P_i^{10}$ \\
\hline
$a_1$ & $a_2$ & $a_2$ & $a_3$ & $a_3$ & $a_3$ & $a_4$ & $a_4$ & $a_4$ & $a_1$ \\
$a_2$ & $a_1$ & $a_3$ & $a_2$ & $a_2$ & $a_4$ & $a_3$ & $a_3$ & $a_1$ & $a_4$ \\
$a_3$ & $a_3$ & $a_1$ & $a_1$ & $a_4$ & $a_2$ & $a_2$ & $a_1$ & $a_3$ & $a_3$ \\
$a_4$ & $a_4$ & $a_4$ & $a_4$ & $a_1$ & $a_1$ & $a_1$ & $a_2$ & $a_2$ & $a_2$
\end{tabular}
\caption{The domain $\bar{\mathbb{D}}$}
\label{T7}
\end{table}
 
Readers can verify that $\bar{\mathbb{D}}$ is connected and satisfies the disagreement property, but does not satisfy the connected with distinct neighbours property (see Figure \ref{fig:ex7}). \hfill $\blacksquare$
\end{example}

\begin{figure}[h!]
\centering
\includegraphics[width=0.7\textwidth]{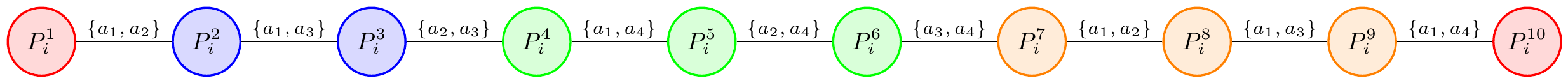}
\caption{Paths induced by the adjacent preferences in $\bar{\mathbb{D}}$}
\label{fig:ex7}
\end{figure}

While our necessary condition is stronger than that of \cite{hong2023unanimity}, our sufficiency result covers a broader class of domains. Every SCD domain is L-tops-only (Proposition 1 of \cite{hong2023unanimity}), so on minimally rich SCD domains Remark \ref{R4} says that dictatorship holds exactly when the domain is connected with distinct neighbours. Combined with Theorem 1 of \cite{hong2023unanimity}, this makes our property equivalent to the disagreement property on such domains. However, not every L-tops-only domain is SCD: by Remark \ref{R6}, $\mathbb{D}^*$ belongs to $\mathbb{D}^{LDICT}$, and hence is L-tops-only, but it is not SCD, so Theorem 1 of \cite{hong2023unanimity} does not apply to it.

\section{Conclusion} \label{sec6}

This paper studies preference domains on which every unanimous and locally strategy-proof scf must be dictatorial. We introduce a new domain condition, connected with distinct neighbours, and establish it as a necessary condition for dictatorship under unanimity and local strategy-proofness. Moreover, we show that this condition is also sufficient when restricted to L-tops-only domains, leading to a precise characterization: the set of minimally rich domains on which every unanimous and locally strategy-proof scf satisfies dictatorship is exactly the intersection of connected with distinct neighbours domains and L-tops-only domains. Further, we prove that on every connected with distinct neighbours domain, whether or not it is L-tops-only, unanimity and strategy-proofness (a stronger requirement) guarantee dictatorship, contributing to the literature on dictatorial and tops-only domains. Figure \ref{F2} summarizes these findings, illustrating the relationships between domains and highlighting the gap between necessity and sufficiency.

\begin{figure}[h!]
    \centering
    \includegraphics[width=.9\textwidth]{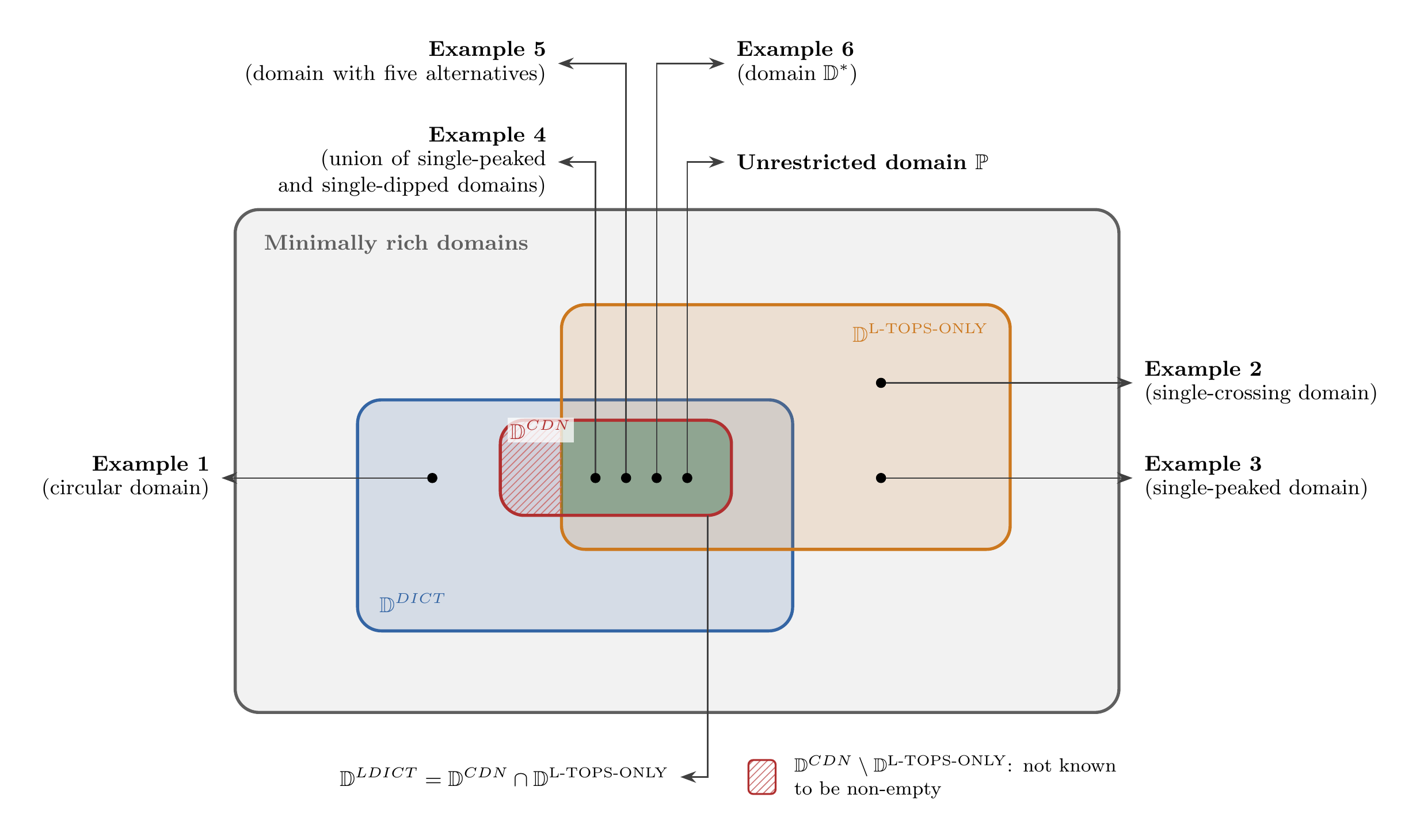}
    \caption{A Summary of Results}
    \label{F2}
\end{figure}

Our necessary condition strengthens existing results (e.g., \cite{hong2023unanimity}), while our sufficiency result applies to domains outside the SCD class (e.g., Example \ref{ex6}). We also provide new insights into when local strategy-proofness implies strategy-proofness for unanimous scfs, extending prior work (\cite{kumar2021localb} and \cite{hong2023unanimity}). Whether $\mathbb{D}^{LDICT}=\mathbb{D}^{CDN}$ remains open, which hinges on whether connected with distinct neighbours domains are always L-tops-only. We conjecture that this inclusion holds, which would close the gap and provide a complete characterization. Future work could explore these results for random scfs, as well as extend our analysis to more general network-based notions of localness introduced in \cite{kumar2021local}.

\bibliographystyle{ecta}
\bibliography{order}
\newpage
\appendix
\section*{Appendix} \label{Sec:app}

\subsection*{1. The Proof of Theorem \ref{thm:1}}
\begin{proof} Let $\mathbb{D}$ be a minimally rich domain. Suppose, any unanimous and local strategy-proof scf $f: \mathbb{D}^n\rightarrow A$ satisfies dictatorship. We will show that $\mathbb{D}$ is connected with distinct neighbours. We assume that $\mathbb{D}$ is not connected with distinct neighbours. In what follows, either $(i)$ $\mathbb{D}$ is not connected or $(ii)$ $\mathbb{D}$ is connected but for some $P_i\in \mathbb{D}$, $\mathbb{D}^{TCC}(P_i)$ does not have two distinct neighbours in $\mathbb{D}$. We consider following two cases.

Case $1$: Suppose $\mathbb{D}$ is not connected. We will show that it is possible to define unanimous and locally strategy-proof scf on $\mathbb{D}$ that does not satisfy dictatorship. Let $P^*_i, P^{**}_i\in \mathbb{D}$ such that $P^*_i$ and $P^{**}_i$ are not connected in $\mathbb{D}$. Let $\bar{\mathbb{D}}$ be the set of all orderings in $\mathbb{D}$ that are connected to $P^*_i$ i.e., $\bar{\mathbb{D}}=\{P''_i\in \mathbb{D}: P''_i \textit{ and } P^*_i \textit{ are connected in } \mathbb{D}\}\cup P^*_i$. 

Fix voters $1$ and $2$. Define $f^{1,2}:\mathbb{D}^n\rightarrow A$ as follows. For any $P\in \mathbb{D}^n$
	
		 \[ f^{1,2}(P)=
		\begin{cases}
			r_1(P_1) & \mbox{if } P_1\in \bar{\mathbb{D}} \\
			r_1(P_2) & \mbox{if } P_1\notin \bar{\mathbb{D}}.
		\end{cases}
		 \]
		
It is easy to verify that $f$ is unanimous and it does not satisfy dictatorship. We show that $f^{1,2}$ is locally strategy-proof. Consider an arbitrary preference profile $P\in \mathbb{D}^n$. Note that voter $j$, $j\in \{3,\ldots ,n\}$, cannot change the outcome by unilaterally deviating at $P$. Therefore, she cannot locally manipulate. Similarly, voter 2 cannot change the outcome by unilaterally deviating at $P$ if $P_1\in \bar{\mathbb{D}}$. If $P_1\notin \bar{\mathbb{D}}$, the outcome is voter 2's top ranked alternative. Therefore, voter 2 cannot locally manipulate at $P$. Finally, we will show that voter 1 cannot locally manipulate at $P$ as well. Note that for any $P'_i\in \bar{\mathbb{D}}$, if $P''_i\in \mathbb{D}$ is adjacent to $P'_i$, then  $P''_i\in \bar{\mathbb{D}}$. Similarly, for any $P'_i\notin \bar{\mathbb{D}}$, if $P''_i\in \mathbb{D}$ is adjacent to $P'_i$, then  $P''_i\notin \bar{\mathbb{D}}$. If $P_1\in \bar{\mathbb{D}}$, the outcome is voter 1's top ranked alternative. Therefore, voter 1 cannot locally manipulate at $P$ whenever $P_1\in \bar{\mathbb{D}}$. If $P_1\notin \bar{\mathbb{D}}$, then for any $P'_1$ such that $P_1$ and $P'_1$ are adjacent, we have $f^{1,2}(P)= f^{1,2}(P'_1,P_{-1})= r_1(P_2)$. Hence, voter 1 cannot locally manipulate at $P$.
	
Case $2$: Suppose $\mathbb{D}$ is connected and $P^*_i\in \mathbb{D}$ such that $\mathbb{D}^{TCC}(P^*_i)$ does not have two distinct neighbours in $\mathbb{D}$.  Since $\mathbb{D}$ is minimally rich and connected, $N(\mathbb{D}^{TCC}(P^*_i), \mathbb{D})$ is a non-empty set. Moreover, since $\mathbb{D}^{TCC}(P^*_i)$ does not have two distinct neighbours, for any $P'_i, P''_i \in N(\mathbb{D}^{TCC}(P^*_i), \mathbb{D})$, we have $r_1(P'_i)=r_1(P''_i)$. We assume that $r_1(P^*_i)=a$ and $r_1(P'_i)=b$ where $P'_i\in N(\mathbb{D}^{TCC}(P^*_i), \mathbb{D})$. 

Fix voters $1$ and $2$. Define $f^{1,2}:\mathbb{D}^n\rightarrow A$ as follows. For any $P\in \mathbb{D}^n$
	
		 \[ f^{1,2}(P)=
		\begin{cases}
			a & \mbox{if } P_1\in \mathbb{D}^{TCC}(P^*_i) \textit{ and } a\;P_2\;b \\
			b & \mbox{if } P_1\in \mathbb{D}^{TCC}(P^*_i) \textit{ and } b\;P_2\;a  \\
			r_1(P_1) & \mbox{if } P_1\notin \mathbb{D}^{TCC}(P^*_i) .
		\end{cases}
		 \]

It can be verified that $f^{1,2}$ is unanimous and it does not satisfy dictatorship. We show that $f^{1,2}$ is locally strategy-proof as well. Consider an arbitrary preference profile $P\in \mathbb{D}^n$. Note that voter $j$, $j\in \{3,\ldots ,n\}$, cannot change the outcome by unilaterally deviating at $P$. Therefore, she cannot locally manipulate. Similarly, if $P_1\notin \mathbb{D}^{TCC}(P^*_i)$, voter 2 cannot change the outcome by unilateral deviation. If $P_1 \in \mathbb{D}^{TCC}(P^*_i)$, the outcome is voter 2's preferred alternative between $a$ and $b$. Therefore, voter 2 cannot locally manipulate at $P$. Finally, we will show that voter 1 cannot locally manipulate at $P$ as well. If $P_1\notin \mathbb{D}^{TCC}(P^*_i)$, the outcome is voter 1's top ranked alternative. Therefore, voter 1 does not have any incentive to manipulate. If $P_1\in \mathbb{D}^{TCC}(P^*_i)$, then for any $P'_1$ such that $P_1$ and $P'_1$ are adjacent, we have  either $f^{1,2}(P)=a\;P_1\; b=f^{1,2}(P'_1,P_{-1})$ or $f^{1,2}(P)= f^{1,2}(P'_1,P_{-1})$. Hence, voter 1 cannot locally manipulate at $P$. 
\end{proof}

\subsection*{2. The Proof of Theorem \ref{thm:2} and Theorem \ref{thm:3}}

\subsubsection*{2.1 Foundational Lemma}

The proof of Theorem \ref{thm:2} relies on standard concepts from graph theory. We will introduce these concepts to begin with. An \textbf{undirected graph} $G$ is defined by $(N, E)$ where $N$ is a finite set and $E$ is a collection of unordered pairs of distinct elements from $N$.  The elements of $N$ are called \textbf{vertices} or \textbf{nodes} of graph $G$. The elements of $E$ are called \textbf{edges} of graph $G$. A \textbf{path} is a sequence of distinct vertices $(a^1,\ldots, a^k)$ such that $(a^j, a^{j+1})\in E$ for all $1 \leq j < k$. The path $(a^1,\ldots, a^k)$ is called a path from $a^1$ to $a^k$. A graph is \textbf{connected} if there is a path between every pair of vertices. A \textbf{cycle} is a sequence of vertices $(a^1,\ldots, a^k, a^{k+1})$ with $k > 2$ such that $(a^j, a^{j+1})\in E$ for all $1 \leq j \leq k$, $(a^1,\ldots, a^k)$ is a path, and $a^1 = a^{k+1}$. If $(a^i, a^j)\in E$, then $a^i$ and $a^j$ are called \textbf{end points} of this edge. The \textbf{degree} of a vertex is the number of edges for which that vertex is an end point. So, for every $a^i \in A$, we have $deg(a^i)=|\{a^j\in A : (a^i, a^j)\in E\}|$. A vertex $a^l$ belongs to a cycle $(a^1,\ldots, a^k, a^{k+1})$ if $a^l\in \{a^1,\ldots, a^k\}$. Similarly,  a vertex $a^l$ belongs to a path $(a^1,\ldots, a^k)$ if $a^l\in \{a^1,\ldots, a^k\}$. Two paths (cycles) are distinct if no vertex belongs to both paths (cycles). The \textbf{length} of a path (cycle) is the number of edges in a path (cycle).

For any $\mathbb{D}\subseteq \mathbb{P}$, we denote the undirected graph induced by $\mathbb{D}$, $G(\mathbb{D})=(A,E)$ as follows: 
 
\begin{itemize}
\item The set of nodes or vertices is the set of alternatives $A$. 
\item An edge $(a,b)\in E$ if and only if there exist $P_i,P_i'\in \mathbb{D}$ such that $P_i\sim P_i'$, $a=r_1(P_i)=r_2(P_i')$ and $b=r_1(P_i')=r_2(P_i)$. 
	\end{itemize}
	
The following lemma introduces some properties of $G(\mathbb{D})$ when $\mathbb{D}$ is minimally rich and connected with distinct neighbours.

\begin{lemma}\label{L1} Let $\mathbb{D}$ be minimally rich and connected with distinct neighbours. Then, $G(\mathbb{D})$ satisfies the following properties.
\begin{enumerate}
\item[(a)] $G(\mathbb{D})$ is connected.
\item[(b)] For every $a\in A$, $deg(a)\geq 2$.
\item[(c)] $G(\mathbb{D})$ contains a cycle.
\item[(d)] For every $a\in A$, if $a$ does not belong to a cycle in $G(\mathbb{D})$, then there exists a path  $(a^1,\ldots,a^k)$, $k\geq 3$, such that $a\in \{a^2,\ldots,a^{k-1}\}$, and $a^1$ and $a^k$ belong to distinct cycles.
\end{enumerate}
\end{lemma}

\begin{proof} Let $\mathbb{D}$ be minimally rich and connected with distinct neighbours. 
\begin{enumerate}
\item[(a)] We show that $G(\mathbb{D})$ is connected. Consider a pair of of vertices $a$ and $b$ in $G(\mathbb{D})$. We will show that there is a path between $a$ and $b$. Since $\mathbb{D}$ is minimally rich, there exist $P_i, P'_i\in \mathbb{D}$ such that $r_1(P_i)=a$ and $r_1(P'_i)=b$. Since $\mathbb{D}$ is connected, there exists a path from $P_i$ to $P'_i$ in $\mathbb{D}$. Let $(P_i=P^1_i,P^2_i,\ldots,P^l_i=P'_i)$ be a path from $P_i$ to $P'_i$ in $\mathbb{D}$.   Following steps will ensure the existence of a path from $a$ to $b$ in $G(\mathbb{D})$. 
\begin{itemize}
\item[Step $1$:] We denote \[j^{1} = \max_{j\in \{1,\ldots,l\}} \{j: r_1(P^j_i)=a\},\]

If $r_1(P^{j^{1}+1}_i)=b$, then $(a,b)$ is a path from $a$ to $b$ and  we are done. Otherwise, we move to step $2$.

\item[Step $2$:] We denote \[j^{2} = \max_{j\in \{j^{1}+1,\ldots,l\}} \{j: r_1(P^j_i)=r_1(P^{j^{1}+1}_i)\},\]

If $r_1(P^{j^{2}+1}_i)=b$, then $(a,r_1(P^{j^{2}}_i),b)$ is a path from $a$ to $b$ and  we are done. Otherwise, we move to step $3$.

\item[\vdots]

\item[Step $t$:] We denote \[j^{t} = \max_{j\in \{j^{t-1}+1,\ldots,l\}} \{j: r_1(P^j_i)=r_1(P^{j^{t-1}+1}_i)\},\]

If $r_1(P^{j^{t}+1}_i)=b$, then $(a,r_1(P^{j^{2}}_i),\ldots,r_1(P^{j^{t}}_i),b)$ is a path from $a$ to $b$ and  we are done. Otherwise, we move to step $t+1$. 
\end{itemize} 

Since $r_1(P^l_i)=b$, in fewer than $l$ steps, we will find a path from $a$ to $b$.

\item[(b)] Consider any $a\in A$. We will show that $deg(a)\geq 2$. Since $\mathbb{D}$ is minimally rich, there exists $P_i\in \mathbb{D}$ such that $r_1(P_i)=a$. Since $\mathbb{D}$ is connected with distinct neighbours, there exist $P'_i, P''_i\in N(\mathbb{D}^{TC}(P_i),\mathbb{D})$  such that $r_1(P'_i)\neq r_1(P''_i)$. Since $(a,r_1(P'_i))$ and $(a,r_1(P''_i))$ are two distinct edges in $G(\mathbb{D})$, we have that $deg(a)\geq 2$.

\item[(c)] We assume, for contradiction, that $G(\mathbb{D})$ does not contain a cycle. Since $G(\mathbb{D})$ is connected (by Lemma \ref{L1}(a)), it is well-known that $G(\mathbb{D})$ is a tree graph and there exists a vertex with degree $1$. This contradicts Lemma \ref{L1}(b). Therefore, $G(\mathbb{D})$ contains a cycle.

\item[(d)] Suppose $a \in A$ does not belong to a cycle in $G(\mathbb{D})$. Let $P(a)$ be the set of all paths that contain $a$. Note that $P(a)$ is finite. Let $(a^1, \ldots, a^k)$ be a path that has the maximum length among all paths in $P(a)$. We assume that $a = a^j$, where $j \in \{1, \ldots, k\}$. First, we show that $(a^1, \ldots, a^i, a^1)$ is a cycle in $G(\mathbb{D})$ where $2 < i < j$. Since $\text{deg}(a^1) \geq 2$, there exists $b \in A$ such that $b \neq a^2$ and $(a^1, b)$ is an edge. Since $(a^1, \ldots, a^k)$ has the maximum length, we have that $b \in \{a^3, \ldots, a^k\}$. Without loss of generality, we assume that $b = a^i$. Therefore, $(a^1, \ldots, a^i, a^1)$ is a cycle in $G(\mathbb{D})$. Moreover, since $a^j$ does not belong to a cycle, it must be the case that $i < j$. A similar argument will establish that $(a^k, a^{k-1}, \ldots, a^{i'}, a^k)$ is a cycle in $G(\mathbb{D})$ where $j < i' < k - 1$. Therefore,  $(a^i,a^{i+1}, \ldots, a^{i'-1},a^{i'})$ is a path where $a\in \{a^{i+1}, \ldots, a^{i'-1}\}$. Also, cycles $(a^1, \ldots, a^i, a^1)$ and $(a^k, a^{k-1}, \ldots, a^{i'}, a^k)$ are distinct. 
\end{enumerate}
\end{proof}

\subsubsection*{2.2 The Proof of Theorem \ref{thm:2}}

\begin{prop}\label{P1} Let $\mathbb{D}$ be a minimally rich and connected with distinct neighbours domain. If $f: \mathbb{D}^n\rightarrow A$ is unanimous, tops-only and local strategy-proof, then $f$ satisfies dictatorship.
\end{prop}

\begin{proof}
Let $\mathbb{D}$ be a minimally rich and connected with distinct neighbours domain. We assume that $f: \mathbb{D}^n\rightarrow A$ satisfies unanimity, tops-onlyness and local strategy-proofness.  We show that $f$ satisfies dictatorship. 

For any $a\in A$, we say that voter $i$ is decisive over $a$ if for all $P\in \mathbb{D}^n$ such that $r_1(P_i)=a$, we have $f(P)=a$. Note that voter $i$ is a dictator if she is decisive over all $a\in A$.  

We prove the proposition using an induction argument on the number of voters. Let $N=\{1,2\}$. We show that $f: \mathbb{D}^2\rightarrow A$ satisfies dictatorship. Let $G(\mathbb{D})=(A,E)$ be the undirected graph induced by $\mathbb{D}$. 

\begin{claim}\label{C1} Let $(a,b)\in E$. Then, either
\begin{enumerate}
\item[(i)] for all $(P_1,P_2)\in \mathbb{D}^2$ such that $r_1(P_1)=a$ and $r_1(P_2)=b$, $f(P_1,P_2)=a$, or
\item[(ii)] for all $(P_1,P_2)\in \mathbb{D}^2$ such that $r_1(P_1)=a$ and $r_1(P_2)=b$, $f(P_1,P_2)=b$.
\end{enumerate}
\end{claim}

\begin{proof} Let $(\bar{P_1},\bar{P_2})\in \mathbb{D}^2$ such that $r_1(\bar{P_1})=r_2(\bar{P_2})=a$ and $r_1(\bar{P_2})=r_2(\bar{P_1})=b$ and $\bar{P_1}\sim \bar{P_2}$. Note that such a profile exists because $(a,b)\in E$. We show that $f(\bar{P_1},\bar{P_2})$ is either $a$ or $b$. Suppose not, i.e., $f(\bar{P_1},\bar{P_2})=c \neq a,b$. Then, voter $1$ can locally manipulate at $(\bar{P_1},\bar{P_2})$ via $\bar{P_2}$ and obtain the outcome $b$ by unanimity. Therefore, $f(\bar{P_1},\bar{P_2})$ is either $a$ or $b$.  

If $f(\bar{P_1},\bar{P_2})=a$, then by tops-onlyness, for all $(P_1,P_2)\in \mathbb{D}^2$ such that $r_1(P_1)=a$ and $r_1(P_2)=b$, we have $f(P_1,P_2)=a$. Similarly, if $f(\bar{P_1},\bar{P_2})=b$, then tops-onlyness implies that for all $(P_1,P_2)\in \mathbb{D}^2$ such that $r_1(P_1)=a$ and $r_1(P_2)=b$, we have $f(P_1,P_2)=b$.
\end{proof}

\begin{claim}\label{C2} Let $(a^1,\ldots,a^k)$ be a path in $G(\mathbb{D})$ with $k\geq 3$. Suppose, for all $(P_1,P_2)\in \mathbb{D}^2$ such that $r_1(P_1)=a^1$ and $r_1(P_2)=a^2$, we have $f(P_1,P_2)=a^1$. Then, it follows that for all $(P_1,P_2)\in \mathbb{D}^2$ such that $r_1(P_1)=a^{k-1}$ and $r_1(P_2)=a^k$, we have $f(P_1,P_2)=a^{k-1}$.
\end{claim}

\begin{proof} Let $(a^1,\ldots,a^k)$ be a path in $G(\mathbb{D})$ with $k\geq 3$. Suppose that for all $(P_1,P_2)\in \mathbb{D}^2$ such that $r_1(P_1)=a^1$ and $r_1(P_2)=a^2$, we have $f(P_1,P_2)=a^1$. We will prove, by induction on $l\in \{1,2,\ldots,k-1\}$, that for all $(P_1,P_2)\in \mathbb{D}^2$ such that $r_1(P_1)=a^l$ and $r_1(P_2)=a^{l+1}$, we have $f(P_1,P_2)=a^l$.

The base case i.e., $l=1$ is given by assumption. We assume that for some $l\in \{1,2,\ldots,k-2\}$, this is true. We will show it holds for $l+1$, i.e., for all $(P_1,P_2)\in \mathbb{D}^2$ such that $r_1(P_1)=a^{l+1}$, $r_1(P_2)=a^{l+2}$, we have $f(P_1,P_2)=a^{l+1}$.

Consider a preference profile $(P'_1,P'_2)$ and orderings $P''_1$ and $P''_2$ such that $(i)$ $r_1(P'_1)=a^{l}$, $r_1(P''_1)=a^{l+1}$ and $P'_1\sim P''_1$, and $(ii)$ $r_1(P'_2)=a^{l+1}$, $r_1(P''_2)=a^{l+2}$ and $P'_2\sim P''_2$. By our assumption, we have that $f(P'_1,P'_2)=a^{l}$. By local strategy-proofness, we have that $f(P'_1,P''_2)=a^{l}$. Again by local strategy-proofness, $f(P''_1,P''_2)\in \{a^l,a^{l+1}\}$. Since $(a^{l+1},a^{l+2})\in E$, by claim \ref{C1}, we have that $f(P''_1,P''_2)=a^{l+1}$ and  for all $(P_1,P_2)\in \mathbb{D}^2$ such that $r_1(P_1)=a^{l+1}$ and $r_1(P_2)=a^{l+2}$, $f(P_1,P_2)=a^{l+1}$.

Therefore, it follows that for all $(P_1,P_2)\in \mathbb{D}^2$ such that $r_1(P_1)=a^{k-1}$ and $r_1(P_2)=a^k$, we have $f(P_1,P_2)=a^{k-1}$.
\end{proof}

\begin{claim}\label{C3} Let $(a^1,\ldots,a^k)$ be a path in $G(\mathbb{D})$ with $k\geq 3$. Suppose, for all $(P_1,P_2)\in \mathbb{D}^2$ such that $r_1(P_1)=a^1$ and $r_1(P_2)=a^2$, we have $f(P_1,P_2)=a^1$. Then, it follows that for all $(P_1,P_2)\in \mathbb{D}^2$ such that $r_1(P_1)=a^1$ and $r_1(P_2)=a^k$, we have $f(P_1,P_2)=a^1$. 
\end{claim}

\begin{proof} Let $(a^1,\ldots,a^k)$ be a path in $G(\mathbb{D})$. We assume that for all $(P_1,P_2)\in \mathbb{D}^2$ such that $r_1(P_1)=a^1$ and $r_1(P_2)=a^2$, we have $f(P_1,P_2)=a^1$. We will prove, by induction on $l\in \{2,\ldots,k\}$, that for all $(P_1,P_2)\in \mathbb{D}^2$ such that $r_1(P_1)=a^1$ and $r_1(P_2)=a^{l}$, we have $f(P_1,P_2)=a^1$.

The base case i.e., $l=2$ is given by assumption. We assume that for some $l\in \{2,\ldots,k-1\}$, this is true. We will show it holds for $l+1$, i.e., for all $(P_1,P_2)\in \mathbb{D}^2$ such that $r_1(P_1)=a^{1}$, $r_1(P_2)=a^{l+1}$, we have $f(P_1,P_2)=a^{1}$.

Consider a preference profile $(P'_1,P'_2)$ and an ordering $P''_2$ such that $(i)$ $r_1(P'_1)=a^1$ and $(ii)$ $r_1(P'_2)=a^{l}$, $r_1(P''_2)=a^{l+1}$ and $P'_2\sim P''_2$. By our assumption, we have that $f(P'_1,P'_2)=a^1$. By local strategy-proofness, we have that $f(P'_1,P''_2)=a^1$. By tops-onlyness, we conclude that for all $(P_1,P_2)\in \mathbb{D}^2$ such that $r_1(P_1)=a^1$ and $r_1(P_2)=a^{l+1}$, $f(P_1,P_2)=a^1$.

Therefore, it follows that for all $(P_1,P_2)\in \mathbb{D}^2$ such that $r_1(P_1)=a^1$ and $r_1(P_2)=a^k$, we have $f(P_1,P_2)=a^1$. 
\end{proof}

We show the following lemmas.

\begin{lemma}\label{L2} For any distinct $a$ and $b$ in $A$, it is impossible that voter $1$ is decisive over $a$ and voter $2$ is decisive over $b$.
\end{lemma}
\begin{proof}
Let $(P_1,P_2)\in \mathbb{D}^2$ such that $r_1(P_1)=a$ and $r_1(P_2)=b$. Since $f$ is a function, it is impossible for 
$f(P_1,P_2)$ to equal both $a$ and $b$ simultaneously. Therefore, it is impossible that voter $1$ is decisive over $a$ and voter $2$ is decisive over $b$.
\end{proof}

\begin{lemma}\label{L3} Let $(a,b)\in E$. Then, either voter $1$ is decisive over $a$ or voter $2$ is decisive over $b$. 
\end{lemma}
\begin{proof}
Fix $(a,b)\in E$. Let $(\bar{P_1},\bar{P_2})\in \mathbb{D}^2$ such that $r_1(\bar{P_1})=a$ and $r_1(\bar{P_2})=b$. By Claim \ref{C1}, $f(\bar{P_1},\bar{P_2})$ is either $a$ or $b$. We will show that if $f(\bar{P_1},\bar{P_2})=a$, then voter $1$ is decisive over $a$. A similar argument will establish that if $f(\bar{P_1},\bar{P_2})=b$, then voter $2$ is decisive over $b$. W.l.o.g., we assume that $f(\bar{P_1},\bar{P_2})=a$. We will show that for all $(P_1,P_2)\in \mathbb{D}^2$ such that $r_1(P_1)=a$, we have $f(P_1,P_2)=a$. Following three claims will establish the lemma.

\begin{claim}\label{C4} For all $(P_1,P_2)\in \mathbb{D}^2$ such that $r_1(P_1)=r_1(P_2)=a$, we have $f(P_1,P_2)=a$.
\end{claim}
\begin{proof}
This follows from the fact the $f$ is unanimous.
\end{proof}

\begin{claim}\label{C5} For all $(P_1,P_2)\in \mathbb{D}^2$ such that $r_1(P_1)=a$ and $r_1(P_2)=b$, we have $f(P_1,P_2)=a$.
\end{claim}
\begin{proof}
Note that $f(\bar{P_1},\bar{P_2})=a$. By claim \ref{C1}, for all $(P_1,P_2)\in \mathbb{D}^2$ such that $r_1(P_1)=a$ and $r_1(P_2)=b$, we have $f(P_1,P_2)=a$.
\end{proof}

\begin{claim}\label{C6} For all $(P_1,P_2)\in \mathbb{D}^2$ such that $r_1(P_1)=a$, we have $f(P_1,P_2)=a$.
\end{claim}
\begin{proof}
Consider a profile $(P_1,P_2)\in \mathbb{D}^2$ such that $r_1(P_1)=a$ and $r_1(P_2)=c$. If $c=a$, we are done by Claim \ref{C4}. If $c=b$, we are done by Claim \ref{C5}. Therefore, we consider the case where $c\neq a,b$.

Suppose there exists a path from $a$ to $c$ in $G(\mathbb{D})$ that contains $b$. Since $(a,b)\in E$, there exists a path $(a^1 = a, a^2, \ldots, a^k = c)$ in $G(\mathbb{D})$ such that $a^2 = b$. Then, by Claim \ref{C3} and Claim \ref{C5}, it follows that $f(P_1, P_2) = a$.

We only need to consider the case that all paths from $a$ to $c$ in $G(\mathbb{D})$ do not contain $b$. Let $(a^1=a,a^2,\ldots,a^k=c)$ be a path that does not contain $b$. We will consider following cases to complete the proof.

Case 1: $b$ does not belong to a cycle in $G(\mathbb{D})$. By Lemma \ref{L1}.(d), there exist a path $(b^1,b^2,\ldots,b^p)$ and a cycle $(c^1,c^2,\ldots,c^k,c^{k+1})$ in $G(\mathbb{D})$ such that $b^1=b$ and $b^p=c^1$. Now consider the path $(a^1=a,b^1=b,\ldots,b^p=c^1,c^2,\ldots,c^{k-1},c^k)$. By claim \ref{C2} and claim \ref{C5}, for all $(P_1,P_2)\in \mathbb{D}^2$ such that $r_1(P_1)=c^{k-1}$, $r_1(P_2)=c^k$, we have $f(P_1,P_2)=c^{k-1}$.

 Next we consider the path $(c^{k-1},c^k,c^1=b^p,b^{p-1},\ldots,b^2,b^1=b,a^1=a,a^2)$. By the result from the previous paragraph and Claim \ref{C2}, we conclude that for all $(P_1,P_2)\in \mathbb{D}^2$ such that $r_1(P_1)=a$, $r_1(P_2)=a^2$, we have $f(P_1,P_2)=a$.  
 
 Finally consider the path $(a^1=a,a^2,\ldots,a^k=c)$. By the result from the previous paragraph and Claim \ref{C3}, we conclude that for all $(P_1,P_2)\in \mathbb{D}^2$ such that $r_1(P_1)=a$, $r_1(P_2)=c$, we have $f(P_1,P_2)=a$.

Case 2: $b$ belongs to a cycle in $G(\mathbb{D})$. Let $(c^1,c^2,\ldots,c^k,c^{k+1})$ be the cycle and $c^1=b$. we consider following two sub cases.

Sub-case 2.1: $a$ belongs to the cycle $(c^1,c^2,\ldots,c^k,c^{k+1})$. W.l.o.g., we assume that $c^k=a$. Now consider the path $(c^k,c^1,c^2,\ldots,c^{k-2},c^{k-1})$. By claim \ref{C2} and claim \ref{C5}, for all $(P_1,P_2)\in \mathbb{D}^2$ such that $r_1(P_1)=c^{k-2}$, $r_1(P_2)=c^{k-1}$, we have $f(P_1,P_2)=c^{k-2}$.

Next, consider the path $(c^{k-2},c^{k-1},c^k=a^1=a,a^2)$. By the result from the previous paragraph and Claim \ref{C2}, we have that for all $(P_1,P_2)\in \mathbb{D}^2$ such that $r_1(P_1)=a$, $r_1(P_2)=a^2$, $f(P_1,P_2)=a$.

Finally we consider the path $(a^1=a,a^2,\ldots,a^k=c)$.  By the result from the previous paragraph and Claim \ref{C3}, we have that for all $(P_1,P_2)\in \mathbb{D}^2$ such that $r_1(P_1)=a$, $r_1(P_2)=c$, $f(P_1,P_2)=a$.

Sub-case 2.2: $a$ does not belong to the cycle $(c^1,c^2,\ldots,c^k,c^{k+1})$. Now consider the path $(a^1=a,c^1=b,c^2,\ldots,c^{k-1},c^k)$. By claim \ref{C5} and claim \ref{C2}, for all $(P_1,P_2)\in \mathbb{D}^2$ such that $r_1(P_1)=c^{k-1}$, $r_1(P_2)=c^k$, we have $f(P_1,P_2)=c^{k-1}$.

Next, consider the path $(c^{k-1},c^k,c^1,a^1=a,a^2)$. By the result from the previous paragraph and Claim \ref{C2}, we have that for all $(P_1,P_2)\in \mathbb{D}^2$ such that $r_1(P_1)=a$, $r_1(P_2)=a^2$, $f(P_1,P_2)=a$.

Finally we consider the path $(a^1=a,a^2,\ldots,a^k=c)$.  By the result from the previous paragraph and Claim \ref{C3}, we have that for all $(P_1,P_2)\in \mathbb{D}^2$ such that $r_1(P_1)=a$, $r_1(P_2)=c$, $f(P_1,P_2)=a$.

This concludes the proof of the claim.
\end{proof}

Claim \ref{C6} shows that voter $1$ is decisive over $a$ if $f(\bar{P_1},\bar{P_2})=a$. Similarly, if $f(\bar{P_1},\bar{P_2})=b$, a similar argument will establish that voter $2$ is decisive over $b$. This concludes the proof of the lemma.  
\end{proof}

\begin{lemma}\label{L4} For any $a\in A$, either voter $1$ or voter $2$ is decisive over $a$.
\end{lemma}
\begin{proof} Suppose not, i.e., both voter $1$ and voter $2$ are not decisive over $a$. Note that by Lemma \ref{L1}.(b), $deg(a)\geq 2$ in $G(\mathbb{D})$. Therefore, we have two distinct alternatives $b$ and $c$ such that $(b,a)\in E$ and $(a,c)\in E$. Since $(b,a)\in E$, by lemma \ref{L3} and the fact that voter $2$ is not decisive over $a$, we have that voter $1$ is decisive over $b$. Similarly, since $(a,c)\in E$, by lemma \ref{L3} and the fact that voter $1$ is not decisive over $a$, we have that voter $2$ is decisive over $c$. Since voter $1$ is decisive over $b$ and voter $2$ is decisive over $c$, we contradict Lemma \ref{L2}. Therefore, for any $a\in A$, either voter $1$ or voter $2$ is decisive over $a$.   
\end{proof}

\begin{lemma}\label{L5} Either voter $1$ is decisive over all alternatives or voter $2$ is decisive over all alternatives.
\end{lemma}

\begin{proof} Let $a\in A$. By Lemma \ref{L4}, either voter $1$ is decisive over $a$ or voter $2$ is decisive over $a$. First, assume that voter $1$ is decisive over $a$. Let $x\in A\setminus a$. By Lemma \ref{L4}, either voter $1$ is decisive over $x$ or voter $2$ is decisive over $x$. By Lemma \ref{L2}, voter $2$ cannot be decisive over $x$. Therefore, voter $1$ is decisive over $x$. Because $x$ was arbitrary, voter $1$ is decisive over all alternatives. Similarly, if we assume that voter $2$ is decisive over $a$, then voter $2$ is decisive over all alternatives. 
\end{proof}

Therefore, we conclude that $f: \mathbb{D}^2\rightarrow A$ satisfies dictatorship.

Now we are ready to prove the proposition for $N=\{1,2,\ldots,n\}$, $n\geq 3$. We will use induction arguments. Assume that for all integers $k < n $, the following statement is true:

Induction Hypothesis (IH): If $ f : \mathbb{D}^k \rightarrow A$ satisfies unanimity, tops-onlyness and local strategy-proofness, then it satisfies dictatorship. 
	
	Define a scf $g: \mathbb{D}^{n-1} \rightarrow A$ as follows: For all $(P_{1}, P_3,\ldots,P_n) \in \mathbb{D}^{n-1}$,
	
		$$g(P_{1},P_3,\ldots,P_n ) = f(P_1, P_1, P_3,\ldots,P_n).$$
		
Voter $1$ in the scf $g$ is obtained by ``cloning" voters $1$ and $2$ in $f$. Thus if voters
$1$ and $2$ in $f$ have a common ordering $P_1$, then voter $1$ in $g$ has ordering $P_1$. 

\begin{lemma}\label{L6}
The scf $g$ satisfies dictatorship.
\end{lemma}
\begin{proof}
Since $f$ satisfies unanimity and tops-onlyness, it is straightforward that $g$ satisfies unanimity and tops-onlyness. We show that $g$ is locally strategy-proof. It is clear that voters 3 through n cannot locally manipulate in $g$, otherwise they can locally manipulate in $f$. Pick an arbitrary $n-1$ voters profile $(P_1,P_3,\ldots,P_n)$ and let $g(P_1,P_3,\ldots,P_n)=f(P_1,P_1,P_3,\ldots,P_n)=a$. We consider an arbitrary ordering $P'_1$ such that $P'_1\sim P_1$. Let $f(P'_1,P_1,P_3,\ldots,P_n)=b$ and let $f(P'_1,P'_1,P_3,\ldots,P_n)=g(P'_1,P_3,\ldots,P_n)=c$. Since $f$ is locally strategy-proof $a\neq b$ implies $a\;P_1\;b$ and $b\neq c$ implies $b\;P_1\;c$. Since $P_1$ is transitive, $a\neq c$ implies $aP_1 c$. Therefore $g$ cannot be locally manipulated by voter $1$. By IH, we conclude that $g$ satisfies dictatorship.
\end{proof}

Now, we are ready to show that $f: \mathbb{D}^n \rightarrow A$ satisfies dictatorship. First, we will prove the following two claims.  

\begin{claim}\label{C7} Let $(a^1,\ldots,a^k)$ be a path in $G(\mathbb{D})$ with $k\geq 2$. Consider a preference profile $P\in \mathbb{D}^n$ such that $f(P)=r_1(P_i)$ and $r_1(P_j)=a^1$ for some $i,j\in N$ with $i \neq j$. Suppose that $r_1(P_i) \notin \{a^1, \ldots, a^k\}$. Then, for all $P'_j\in \mathbb{D}$ such that $r_1(P'_j)=a^k$, we have $f(P'_j,P_{-j})=r_1(P_i)$.
\end{claim}

\begin{proof} Let $(a^1,\ldots,a^k)$ be a path in $G(\mathbb{D})$. Let $P\in \mathbb{D}^n$ such that $f(P)=r_1(P_i)=a$, $r_1(P_j)=a^1$ and $a\notin \{a^1, \ldots, a^k\}$. Consider a sequence of preference ordering $(P^1_j,P^2_j\ldots,P^k_j)$ such that for all $l\in \{1,2,\ldots,k\}$, $r_1(P^l_j)=a^l$. We will use induction over $l$ to show that $f(P^k_j,P_{-j})=a$.

Suppose $l=1$. Since $f(P)=a$, by tops-onlyness, $f(P^1_j,P_{-j})=a$. We assume that for some $l\in \{1,2,\ldots,k-1\}$, $f(P^l_j,P_{-j})=a$. Now we show that $f(P^{l+1}_j,P_{-j})=a$. Note that $(a^{l},a^{l+1})\in E$. Therefore, we have $P'_j,P''_j\in \mathbb{D}$ such that $r_1(P'_j)=a^{l}$, $r_1(P''_j)=a^{l+1}$ and $P'_j\sim P''_j$. By our assumption, $f(P^{l}_j,P_{-j})=a$. By tops-onlyness, $f(P'_j,P_{-j})=f(P^{l}_j,P_{-j})=a$. By local strategy-proofness, $f(P''_j,P_{-j})=f(P'_j,P_{-j})=a$. Again applying tops-onlyness, we have $f(P^{l+1}_j,P_{-j})=f(P''_j,P_{-j})=a$. Therefore, by induction, we have  $f(P^k_j,P_{-j})=a$.

Finally by tops-onlyness, it follows that for all $P'_j\in \mathbb{D}$ such that $r_1(P'_j)=a^k$, we have $f(P'_j,P_{-j})=a$. This concludes the proof of the claim.
\end{proof}

\begin{claim}\label{C8} Let $(a^1,\ldots,a^k)$ be a path in $G(\mathbb{D})$ with $k\geq 2$. Consider a profile $P\in \mathbb{D}^n$ such that $f(P)=r_1(P_i)=a^k$ and $r_1(P_j)=a^1$ for some $i,j\in N$ with $i \neq j$. Then, for all $P'_j\in \mathbb{D}$ such that $r_1(P'_j)=a^k$, we have $f(P'_j,P_{-j})=r_1(P_i)$.
\end{claim}

\begin{proof} Let $(a^1,\ldots,a^k)$ be a path in $G(\mathbb{D})$. Let $P\in \mathbb{D}^n$ such that $f(P)=r_1(P_i)=a^k$, $r_1(P_j)=a^1$.  Consider the path $(a^1,\ldots,a^{k-1})$ and  orderings $P^*_j$ and $P^{**}_j$ such that $r_1(P^*_j)=a^{k-1}$, $r_1(P^{**}_j)=a^{k}$ and $P^*_j\sim P^{**}_j$. By Claim \ref{C7}, we have that $f(P^*_j,P_{-j})=r_1(P_i)$. By local strategy-proofness, we have that $f(P^{**}_j,P_{-j})=f(P^*_j,P_{-j})=r_1(P_i)=a^k$. Finally applying tops-onlyness, we conclude that for all $P'_j\in \mathbb{D}$ such that $r_1(P'_j)=a^k$, we have $f(P'_j,P_{-j})=r_1(P_i)=a^k$.
\end{proof}

We consider following two cases to show that $f$ satisfies dictatorship.

Case 1: Voter $j\in\{3,4,\ldots,n\}$ is the dictator in $g$. We will show that $j$ is the dictator in $f$. Consider an arbitrary profile $(P_1,P_2,\ldots,P_n)$. Let $r_1(P_j)=a$. We will show that $f(P_1,P_2,\ldots,P_n)=a$. If $r_1(P_1)=r_1(P_2)$, then by tops-onlyness, $f(P_1,P_2,P_3,\ldots,P_n)=f(P_1,P_1,P_3\ldots,P_n)=g(P_1,P_3,\ldots,P_n)=a$. We assume that $r_1(P_1)=b\neq c=r_1(P_2)$. Let $(a^1=b,a^2,\ldots,a^k=c)$ be a path from $b$ to $c$ in $G(\mathbb{D})$. We consider following two sub-cases.

Sub-case 1.1: $(a^1=b,a^2,\ldots,a^k=c)$ does not contain $a$. Since $f(P_1,P_1,P_3\ldots,P_n)=g(P_1,P_3,\ldots,P_n)=a$, by Claim \ref{C7}, we have that $f(P_1,P_2,\ldots,P_n)=a$. 

Sub-case 1.2: $(a^1=b,a^2,\ldots,a^k=c)$ contains $a$. Suppose $a=c$. Consider the path $(a^1=b,a^2,\ldots,a^k=c)$. Since $f(P_1,P_1,P_3\ldots,P_n)=g(P_1,P_3,\ldots,P_n)=a$, by Claim \ref{C8}, we have that $f(P_1,P_2,\ldots,P_n)=a$. 

Suppose $a=b$. Consider the path $(a^k=c,a^{k-1},\ldots,a^1=b)$. Since $f(P_2,P_2,P_3\ldots,P_n)=g(P_2,P_3,\ldots,P_n)=a$, by Claim \ref{C8}, we have that $f(P_1,P_2,\ldots,P_n)=a$. 

Suppose $a=a^l$ where $1<l<k$. Consider the path $(a^1=b,a^2,\ldots,a^l)$ and orderings $P'_2$ and $P''_2$ such that $r_1(P'_2)=a^l=a$, $r_1(P''_2)=a^{l+1}$ and $P'_2\sim P''_2$.  Since $f(P_1,P_1,P_3\ldots,P_n)=g(P_1,P_3,\ldots,P_n)=a$, by Claim \ref{C8}, we have that $f(P_1,P'_2,\ldots,P_n)=a$. Next we show that $f(P_1,P''_2,\ldots,P_n)=a$. By local strategy-proofness, $f(P_1,P''_2,\ldots,P_n)\in \{a,a^{l+1}\}$. Suppose that $f(P_1,P''_2,\ldots,P_n)=a^{l+1}$. Consider the path $(a^1=b,a^2,\ldots,a^l,a^{l+1})$. Since $f(P_1,P''_2,\ldots,P_n)=a^{l+1}$,  by Claim \ref{C8}, we have that $f(P''_2,P''_2,\ldots,P_n)=g(P''_1,P_3,\ldots,P_n)=a^{l+1}$. This contradicts the fact that voter $j$ is the dictator in $g$. Therefore, $f(P_1,P''_2,\ldots,P_n)=a$. Finally, consider the path $(a^{l+1},a^{l+2},\ldots,a^1=b)$. Since $f(P_1,P''_2,\ldots,P_n)=a$, by Claim \ref{C7}, we have that $f(P_1,P_2,\ldots,P_n)=a$.

Case 2: Voter $1$ is the dictator in $g$. We will show that $f$ satisfies dictatorship where voter $i\in \{1,2\}$ is the dictator in $f$. Pick an arbitrary profile of $n-2$ voters, $(P_3,P_4,\ldots,P_n)$. Now define a two-voter scf $h$ as follows: for any $(P_1,P_2)\in \mathbb{D}^2$,
$$h(P_1,P_2)=f(P_1,P_2,P_3,\ldots,P_n).$$ Since voter $1$ is the dictator in $g$ and $f$ is tops-only, it follows that $h$ satisfies unanimity. Moreover, since $f$ is tops-only and locally strategy-proof, it follows immediately that $h$ is tops-only and locally strategy-proof.
From IH, we know that $h$ satisfies dictatorship. To complete the case, we only need to show that the identity of the dictator does not depend on the profile of $n-2$ voters, $(P_3,P_4,\ldots,P_n)$. Suppose that it does depend on this profile. W.l.o.g., we assume that voter $1$ is dictator for $(P_3,P_4,\ldots,P_n)$ while 2 is dictator for $(P'_3,P'_4,\ldots,P'_n)$. Now progressively change preferences for each voter from $3$ through $n$ from the first profile to the second. There must be an
individual $j$ with $3\leq j\leq n$ such that $1$ is the dictator in $(P'_3,\ldots,P'_{j-1},P_j,P_{j+1}\ldots,P_n)$ while $2$ dictates in $(P'_3,\ldots,P'_{j-1},P'_j,P_{j+1}\ldots,P_n)$. Let $r_1(P_j)=a\neq b=r_1(P'_j)$. Pick $P_1$ and $P_2$ such that $r_1(P_1)=b\neq r_1(P_2)$. Let $(a^1=a,a^2,\ldots,a^k=b)$ be a path from $a$ to $b$ in $G(\mathbb{D})$. Note that $f(P_1,P_2,P'_3,\ldots,P'_{j-1},P_j,P_{j+1}\ldots,P_n)=b$. Therefore, by Claim \ref{C8}, we have that $f(P_1,P_2,P'_3,\ldots,P'_{j-1},P'_j,P_{j+1}\ldots,P_n)=b$. This contradicts the fact that $2$ dictates in $(P'_3,\ldots,P'_{j-1},P'_j,P_{j+1}\ldots,P_n)$. Therefore, we conclude that if voter $1$ is the dictator in $g$, then $f$ satisfies dictatorship where voter $i\in \{1,2\}$ is the dictator in $f$. 

Thus, the proof of Proposition \ref{P1} is complete.   
\end{proof}

\begin{proof}{Proof of Theorem \ref{thm:2}.} Let $\mathbb{D}$ be a minimally rich and L-tops-only domain. Consider an unanimous and local strategy-proof scf $f: \mathbb{D}^n\rightarrow A$. Since $\mathbb{D}$ is an L-tops-only domain, $f$ satisfies tops-onlyness.
If $\mathbb{D}$ is connected with distinct neighbours, then, by Proposition \ref{P1}, $f$ satisfies dictatorship. This concludes the proof of Theorem \ref{thm:2}. 
\end{proof}

\subsubsection*{2.3 The Proof of Theorem \ref{thm:3}}

\begin{prop}\label{P2} Let $\mathbb{D}$ be a minimally rich and connected with distinct neighbour domain. If $f: \mathbb{D}^2\rightarrow A$ satisfies unanimity and strategy-proofness, then it satisfies dictatorship.
\end{prop}

\begin{proof} Let $\mathbb{D}$ be a minimally rich and connected with distinct neighbour domain. Let $f: \mathbb{D}^2\rightarrow A$ satisfies unanimity and strategy-proofness. We will show that $f$ satisfies dictatorship.

The proof steps, including the proofs of certain claims and lemmas, are identical to those in the proof of Proposition \ref{P1} for the two-voter case. As some of these proofs are exactly the same, we omit those details here.
 
Let $G(\mathbb{D})= (A,E)$ be the graph induced by $\mathbb{D}$. First we will show following claims.

\begin{claim}\label{C9} let $(a,b)\in E$. Then, either
\begin{enumerate}
\item[(i)] for all $(P_1,P_2)\in \mathbb{D}^2$ such that $r_1(P_1)=a$ and $r_1(P_2)=b$, $f(P_1,P_2)=a$, or
\item[(ii)] for all $(P_1,P_2)\in \mathbb{D}^2$ such that $r_1(P_1)=a$ and $r_1(P_2)=b$, $f(P_1,P_2)=b$.
\end{enumerate}
\end{claim} 

\begin{proof} Let $(\bar{P_1},\bar{P_2})\in \mathbb{D}^2$ such that $r_1(\bar{P_1})=r_2(\bar{P_2})=a$, $r_1(\bar{P_2})=r_1(\bar{P_1})=b$, and $\bar{P_1}\sim \bar{P_2}$. Such a profile exists because $(a,b)\in E$. We show that $f(\bar{P_1},\bar{P_2})$ is either $a$ or $b$. Suppose not, i.e., $f(\bar{P_1},\bar{P_2}) \neq a,b$. Then, voter $1$ can manipulate at $(\bar{P_1},\bar{P_2})$ via a preference ordering where $b$ is first ranked and obtain the outcome $b$ by unanimity. Therefore, $f(\bar{P_1},\bar{P_2})$ is either $a$ or $b$.

Suppose $f(\bar{P_1},\bar{P_2})=a$. Consider an arbitrary profile $(P_1,P_2)\in \mathbb{D}^2$ such that $r_1(P_1)=a$ and $r_1(P_2)=b$. We will show that $f(P_1,P_2)=a$. First we show that $f(\bar{P_1},P_2)\in \{a,b\}$. Suppose $f(\bar{P_1},P_2) \neq a,b$. Then voter $1$ can manipulate at $(\bar{P_1},P_2)$ via a preference ordering where $b$ is first ranked and obtain the outcome $b$ by unanimity. Next we show that $f(\bar{P_1},P_2)\neq b$. If $f(\bar{P_1},P_2)= b$, then voter $2$ can manipulate at $(\bar{P_1},\bar{P_2})$ via $P_2$. Therefore, $f(\bar{P_1},P_2)=a$. 
By strategy-proofness, we have that $f(P_1,P_2)=a$. Since $(P_1,P_2)$ is an arbitrary profile, we conclude that for all $(P_1,P_2)\in \mathbb{D}^2$ such that $r_1(P_1)=a$ and $r_1(P_2)=b$, $f(P_1,P_2)=a$

Suppose $f(\bar{P_1},\bar{P_2})=b$. By using similar arguments as we have used in the previous paragraph, we can establish that for all $(P_1,P_2)\in \mathbb{D}^2$ such that $r_1(P_1)=a$ and $r_1(P_2)=b$, $f(P_1,P_2)=b$. This completes the proof of the claim. 
\end{proof}

\begin{claim}\label{C10} Let $(a^1,\ldots,a^k)$ be a path in $G(\mathbb{D})$ with $k\geq 3$. Suppose, for all $(P_1,P_2)\in \mathbb{D}^2$ such that $r_1(P_1)=a^1$ and $r_1(P_2)=a^2$, we have $f(P_1,P_2)=a^1$. Then, for all $(P_1,P_2)\in \mathbb{D}^2$ such that $r_1(P_1)=a^{k-1}$, $r_1(P_2)=a^k$, we have $f(P_1,P_2)=a^{k-1}$.
\end{claim}

\begin{proof} The proof of claim \ref{C10} mirrors that of claim \ref{C2}, with the only modification being that claim \ref{C1} and local strategy-proofness in Lemma \ref{L4} are replaced by claim \ref{C9} and strategy-proofness, respectively. All other arguments are identical and are thus omitted.
\end{proof}

\begin{claim}\label{C11} Let $(a^1,\ldots,a^k)$ be a path in $G(\mathbb{D})$ with $k\geq 3$. Suppose, for all $(P_1,P_2)\in \mathbb{D}^2$ such that $r_1(P_1)=a^1$ and $r_1(P_2)=a^2$, we have $f(P_1,P_2)=a^1$. Then, for all $(P_1,P_2)\in \mathbb{D}^2$ such that $r_1(P_1)=a^1$ and  $r_1(P_2)=a^k$, we have $f(P_1,P_2)=a^1$.
\end{claim}
\begin{proof} Let $(a^1,\ldots,a^k)$ be a path in $G(\mathbb{D})$. We assume that for all $(P_1,P_2)\in \mathbb{D}^2$ such that $r_1(P_1)=a^1$ and $r_1(P_2)=a^2$, we have $f(P_1,P_2)=a^1$. We will show that for all $(P_1,P_2)\in \mathbb{D}^2$ such that $r_1(P_1)=a^1$ and  $r_1(P_2)=a^k$, we have $f(P_1,P_2)=a^1$.

First we show the following fact.

\begin{fact}\label{F1} Let $(a^1,\ldots,a^k)$ be a path in $G(\mathbb{D})$ with $k\geq 3$. Suppose, 
\begin{enumerate}
\item[(i)] for all $(P_1,P_2)\in \mathbb{D}^2$ such that $r_1(P_1)=a^1$ and $r_1(P_2)=a^2$, we have $f(P_1,P_2)=a^1$ and
\item[(ii)] for all $(P_1,P_2)\in \mathbb{D}^2$ such that $r_1(P_1)=a^2$ and $r_1(P_2)=a^k$, we have $f(P_1,P_2)=a^2$.
\end{enumerate}
 Then, for all $(P_1,P_2)\in \mathbb{D}^2$ such that $r_1(P_1)=a^1$ and $r_1(P_2)=a^k$, we have $f(P_1,P_2)=a^1$.
\end{fact}
\begin{proof} Let $(a^1,a^2,\ldots,a^k)$ be a path in $G(\mathbb{D})$. We assume that the conditions $(i)$ and $(ii)$ in fact 1 hold. Consider an arbitrary profile $(P_1,P_2)\in \mathbb{D}^2$ such that $r_1(P_1)=a^1$ and $r_1(P_2)=a^k$. We will show that $f(P_1,P_2)=a^1$. Consider orderings $P'_1$ and $P''_1$ such that $r_1(P'_1)=a^2$, $r_1(P''_1)=a^1$ and $P'_1\sim P''_1$. By condition $(ii)$, $f(P'_1,P_2)=a^2$. We will show that $f(P''_1,P_2)=a^1$. Since $f(P'_1,P_2)=a^2$, by strategy-proofness, we have  $f(P''_1,P_2)\in \{a^1,a^2\}$. Also note that $f(P''_1,P_2)\neq a^2$. Otherwise, voter 2 can manipulate at a profile where Voter $1$'s preference is $P''_1$ and top-ranked alternative in voter $2$'s preference is $a_2$, via $P_2$. Therefore, $f(P''_1,P_2)=a^1$. Finally, applying strategy-proofness, we have that $f(P_1,P_2)=a^1$. 
\end{proof}

Now we are ready to complete the proof of the claim. Let $l\in \{1,2,\ldots,k-1\}$. We will show that for all $(P_1,P_2)\in \mathbb{D}^2$ such that $r_1(P_1)=a^{k-l}$ and  $r_1(P_2)=a^k$, we have $f(P_1,P_2)=a^{k-l}$. We will use induction arguments to complete the proof.

First we show for the case where $l=1$. By Claim \ref{C10}, we have that for all $(P_1,P_2)\in \mathbb{D}^2$ such that $r_1(P_1)=a^{k-1}$ and  $r_1(P_2)=a^k$, we have $f(P_1,P_2)=a^{k-1}$. 

Next we assume that for some $l\in \{1,2,\ldots,k-2\}$ and for all $(P_1,P_2)\in \mathbb{D}^2$ such that $r_1(P_1)=a^{k-l}$ and  $r_1(P_2)=a^k$, we have $f(P_1,P_2)=a^{k-l}$. We will show that  for all $(P_1,P_2)\in \mathbb{D}^2$ such that $r_1(P_1)=a^{k-(l+1)}$ and  $r_1(P_2)=a^k$, we have $f(P_1,P_2)=a^{k-(l+1)}$. 

If $l=k-2$, we have for all $(P_1,P_2)\in \mathbb{D}^2$ such that $r_1(P_1)=a^2$ and  $r_1(P_2)=a^k$, we have $f(P_1,P_2)=a^2$. Therefore, by fact \ref{F1}, we conclude that for all $(P_1,P_2)\in \mathbb{D}^2$ such that $r_1(P_1)=a^1$ and  $r_1(P_2)=a^k$, $f(P_1,P_2)=a^1$. 

Suppose that $l<k-2$. Consider the path $(a^1,a^2,\ldots,a^{k-(l+1)},a^{k-l})$. By Claim \ref{C10}, we have that for all $(P_1,P_2)\in \mathbb{D}^2$ such that $r_1(P_1)=a^{k-(l+1)}$ and  $r_1(P_2)=a^{k-l}$, we have $f(P_1,P_2)=a^{k-(l+1)}$. By our assumption, we have that for all $(P_1,P_2)\in \mathbb{D}^2$ such that $r_1(P_1)=a^{k-l}$ and  $r_1(P_2)=a^k$, we have $f(P_1,P_2)=a^{k-l}$. Consider the path $(a^{k-(l+1)},a^{k-l},\ldots,a^k)$. Therefore, by fact \ref{F1}, we conclude that for all $(P_1,P_2)\in \mathbb{D}^2$ such that $r_1(P_1)=a^{k-(l+1)}$ and  $r_1(P_2)=a^k$, $f(P_1,P_2)=a^{k-(l+1)}$.

Therefore, by induction, for all $(P_1,P_2)\in \mathbb{D}^2$ such that $r_1(P_1)=a^1$ and  $r_1(P_2)=a^k$, we have $f(P_1,P_2)=a^1$.
This completes the proof of the claim.
\end{proof}

Note that voter $i$ is decisive over $a\in A$ if for all $P\in \mathbb{D}^n$ such that $r_1(P_i)=a$, we have $f(P)=a$. Voter $i$ is a dictator if she is decisive over all $a\in A$.

We complete the proof of the proposition by showing following lemmas.

\begin{lemma}\label{L7} For any distinct $a$ and $b$ in $A$, it is impossible that voter $1$ is decisive over $a$ and voter $2$ is decisive over $b$.
\end{lemma}
\begin{proof} The proof of lemma \ref{L7} is identical to the proof of lemma \ref{L2} and is therefore omitted here.
\end{proof}

\begin{lemma}\label{L8} Let $(a,b)\in E$. Then, either voter $1$ is decisive over $a$ or voter $2$ is decisive over $b$. 
\end{lemma}
\begin{proof} The proof of Lemma \ref{L8} is identical to that of Lemma \ref{L3}, with the following modification: Claims \ref{C1}, \ref{C2}, and \ref{C3} in Lemma \ref{L3} are replaced by Claims \ref{C9}, \ref{C10}, and \ref{C11}, respectively. All other arguments remain the same and are therefore omitted.
\end{proof}

\begin{lemma}\label{L9} For any $a\in A$, either voter $1$ or voter $2$ is decisive over $a$.
\end{lemma}
\begin{proof} The proof of Lemma \ref{L9} mirrors that of Lemma \ref{L4}, with the only modification being that Lemmas \ref{L2} and \ref{L3} in Lemma \ref{L4} are replaced by Lemmas \ref{L7} and \ref{L8}, respectively. All other arguments are identical and are thus omitted.
\end{proof}

\begin{lemma}\label{L10} Either voter $1$ is decisive over all alternatives or voter $2$ is decisive over all alternatives.
\end{lemma}
\begin{proof} The proof of Lemma \ref{L10} is identical to that of Lemma \ref{L5}, with the only modification being that Lemmas \ref{L2} and \ref{L4} in Lemma \ref{L5} are replaced by Lemmas \ref{L7} and \ref{L9}, respectively. All other arguments are the same and are thus omitted.
\end{proof}

Therefore, we conclude that $f: \mathbb{D}^2\rightarrow A$ satisfies dictatorship. This concludes the proof of proposition \ref{P2}.
\end{proof}

The following proposition shows that domains where unanimous and strategy-proof scfs satisfy dictatorship with two agents are also domains where unanimous and strategy-proof scfs satisfy dictatorship for an arbitrary number of agents, provided the domain satisfies the minimal richness condition. The proof can be found in Proposition $3.1$ of \cite{Aswal03}, hence omitted.

\begin{prop}(\cite{Aswal03}) \label{P3}
Let $\mathbb{D}$ be a minimally rich domain. Then, the following two statements are equivalent:
\begin{itemize}
 \item[(a)] $f:\mathbb{D}^2\rightarrow A$ is strategy-proof and unanimous $\Rightarrow$ $f$ satisfies dictatorship.
 \item[(b)] If $f:\mathbb{D}^n\rightarrow A$ is strategy-proof and unanimous $\Rightarrow$ $f$ satisfies dictatorship, $n\geq 2$.
\end{itemize}
\end{prop}

\begin{proof}{Proof of Theorem \ref{thm:3}.} Proposition \ref{P3} reduces the problem from an arbitrary number of voters to the two-voter case. The proof of theorem \ref{thm:3} follows directly from Proposition \ref{P2}.
\end{proof}

\subsection*{3. The Proof of Remark \ref{R6} and Remark \ref{R7}}

\begin{proof}{Proof of Remark \ref{R6}.} The domain $\mathbb{D}^*$ is minimally rich, since $P_i^1,P_i^2,P_i^4$ and $P_i^7$ have tops $a_1,a_2,a_3$ and $a_4$, respectively. Let $f:(\mathbb{D}^*)^n\rightarrow A$ be a locally strategy-proof and unanimous scf. First, we will show that $f$ satisfies dictatorship. Consider the following sub-domain $\mathbb{D}=\mathbb{D}^*\setminus \{P_i^{13}\}$. Let $\Bar{f}$ be the restriction of $f$ on $\mathbb{D}$, i.e., $\forall P\in \mathbb{D}^n$, $\Bar{f}(P)=f(P)$. Since $f$ is locally strategy-proof and unanimous, $\Bar{f}$ is also locally strategy-proof and unanimous. Moreover, note that $\mathbb{D}$ (i.e., the domain in Example \ref{ex4}) is minimally rich and connected with distinct neighbours. Further, the same domain appears as Example 5 in \cite{hong2023unanimity}, where it is shown to be SCD, and every SCD domain is L-tops-only by Proposition 1 of \cite{hong2023unanimity}. By Theorem \ref{thm:2}, $\Bar{f}$ satisfies dictatorship. Let $j$ be the dictator in $\Bar{f}$. We will show that $f$ satisfies dictatorship, where agent $j$ is the dictator, i.e.,  $\forall P\in (\mathbb{D}^*)^n$, $f(P)=r_1(P_j)$. 

Consider an arbitrary $P\in (\mathbb{D}^*)^n$. If $P\in \mathbb{D}^n$, then $f(P)=\Bar{f}(P)=r_1(P_j)$. Suppose $P\notin \mathbb{D}^n$. W.l.o.g., we assume that $\exists\; 1\leq m\leq n$ such that $\forall\; i\in \{1,\ldots,m\}$, $P_i=P_i^{13}$ and $\forall\; i>m$, $P_i\ne P_i^{13}$. We consider following cases.

Case 1: $1\leq j\leq m$. Let $\Bar{P}$ be a profile such that $\forall\; i\in \{1,\ldots,m\}$, $\Bar{P_i}=P_i^1$ and $\forall\; i>m$, $\Bar{P_i}=P_i$. Note that $\Bar{P}\in \mathbb{D}^n$. Therefore, $f(\Bar{P})=\Bar{f}(\Bar{P})=r_1(\Bar{P_j})=r_1(P_i^1)=a_1$. By local strategy-proofness 
\begin{align*}
f(P_i^1,\ldots,P_i^1,P_{m+1},\ldots, P_n)= & f(P_1,P_i^1,\ldots,P_i^1,P_{m+1},\ldots, P_n)\\
= & f(P_1,P_2,P_i^1,\ldots,P_i^1,P_{m+1},\ldots, P_n)\\
\vdots\\
=& f(P_1,\ldots,P_m,P_{m+1},\ldots, P_n).
\end{align*}

Each equality holds because $a_1$ is the top-ranked alternative of both $P_i^1$ and $P_i^{13}$, which are adjacent: if the outcome were different after a voter switches from $P_i^1$ to $P_i^{13}$, the voter with true preference $P_i^{13}$ could manipulate by reporting $P_i^1$ and obtaining $a_1$. Therefore, $f(P)=a_1=r_1(P_j)$.

 Case 2: $j>m$. Let $\Bar{P}$ be a profile such that $\forall\; i\in \{1,\ldots,m\}$, $\Bar{P_i}=P_i^1$ and $\forall\; i>m$, $\Bar{P_i}=P_i$. Note that $\Bar{P}\in \mathbb{D}^n$. Therefore, $f(\Bar{P})=\Bar{f}(\Bar{P})=r_1(\Bar{P_j})=r_1(P_j)$. We consider following two sub-cases.

 Sub-Case 2.1: Suppose $r_1(P_j)\in \{a_1,a_3,a_4\}$.  By local strategy-proofness, 
\begin{align*}
f(P_i^1,\ldots,P_i^1,P_{m+1},\ldots, P_n)= & f(P_1,P_i^1,\ldots,P_i^1,P_{m+1},\ldots, P_n)\\
= & f(P_1,P_2,P_i^1,\ldots,P_i^1,P_{m+1},\ldots, P_n)\\
\vdots\\
=& f(P_1,\ldots,P_m,P_{m+1},\ldots, P_n).
\end{align*}
Here $P_i^1$ and $P_i^{13}$ are adjacent and differ only in the relative ranking of $a_2$ and $a_3$. Let $x=r_1(P_j)\in\{a_1,a_3,a_4\}$ be the outcome before a voter switches from $P_i^1$ to $P_i^{13}$, and $y$ the outcome after. Local strategy-proofness gives that $x$ is weakly preferred to $y$ under $P_i^1$ and that $y$ is weakly preferred to $x$ under $P_i^{13}$. For $x\in\{a_1,a_3,a_4\}$ these two conditions force $y=x$. Therefore, $f(P)=r_1(P_j)$.

 Sub-Case 2.2: $r_1(P_j)=a_2$. Consider the profile $\hat{P}$ such that for all $i\ne j$ $\hat{P}_i=P_i$ and $\hat{P}_j=P_i^1$. From sub-case 2.1, $f(\hat{P})=r_1(\hat{P}_j)=r_1(P_i^1)=a_1$. Note that $P_i^1\sim P_i^2$. We will show that $f(P_i^2,\hat{P}_{-j})=r_1(P_i^2)=a_2$. Since $f(\hat{P})=a_1$ and $P_i^1\sim P_i^2$, local strategy-proofness for a voter with true preference $P_i^2$ gives $f(P_i^2,\hat{P}_{-j})\notin\{a_3,a_4\}$.  Suppose  $f(P_i^2,\hat{P}_{-j})=a_1$. Then, by local strategy-proofness,
 \begin{align*}
 f(P_i^2,\hat{P}_{-j})=& f(P_i^1,\hat{P_2}\ldots,\hat{P}_m,\hat{P}_{m+1},\ldots,P_i^2,\ldots,\hat{P}_n)\\
 \vdots\\
=& f(P_i^1,\ldots,P_i^1,\hat{P}_{m+1},\ldots,P_i^2,\ldots,\hat{P}_n)\\
=& \Bar{f}(P_i^1,\ldots,P_i^1,\hat{P}_{m+1},\ldots,P_i^2,\ldots,\hat{P}_n),
 \end{align*}
 (each equality holds because the outcome is $a_1$, the top of both $P_i^1$ and $P_i^{13}$, as in Case 1). This contradicts the fact that $j$ is dictator in $\Bar{f}$. Therefore, $f(P_i^2,\hat{P}_{-j})\ne a_1$. Hence, $f(P_i^2,\hat{P}_{-j})=a_2$. Since $r_1(P_j)=a_2$, we have $P_j\in\{P_i^2,P_i^3\}$. If $P_j=P_i^2$, the equality is immediate. If $P_j=P_i^3$, then $P_i^2\sim P_i^3$ and $a_2$ is the top of both, so local strategy-proofness gives $f(P_j,\hat{P}_{-j})=a_2=r_1(P_j)$. Note that $P=(P_j,\hat{P}_{-j})$. Therefore, $f(P)=r_1(P_j)$.
 
This completes the proof that $f$ satisfies dictatorship. Since $f$ was an arbitrary unanimous and locally strategy-proof scf, $\mathbb{D}^*\in\mathbb{D}^{LDICT}$, and hence $\mathbb{D}^*\in\mathbb{D}^{L-TOPS-ONLY}$. 
\end{proof}

\begin{proof}{Proof of Remark \ref{R7}.} Let $\mathbb{D}$ be a domain that is connected with distinct neighbours. By definition, $\mathbb{D}$ is connected. We show that $\mathbb{D}$ satisfies disagreement property. Consider a pair of alternatives $a$ and $b$ such that $r_1(P_i)=a$ and $r_1(P'_i)=b$ for some adjacent preferences $P_i,P'_i\in \mathbb{D}$. Since $\mathbb{D}$ is a connected with distinct neighbours domain, there exists $\bar{P_i}\in N(\mathbb{D}^{TCC}(P_i),\mathbb{D})$ such that $r_1(\bar{P_i})\notin \{a,b\}$ and $r_2(\bar{P_i})=a$. Similarly, there exists $\hat{P_i}\in N(\mathbb{D}^{TCC}(P'_i),\mathbb{D})$ such that $r_1(\hat{P_i})\notin \{a,b\}$ and $r_2(\hat{P_i})=b$. Note that $a\; \bar{P_i}\; b$ and $b\; \hat{P_i}\; a$. Hence, $\mathbb{D}$ satisfies disagreement property.   
\end{proof}

\end{document}